\documentclass{elsarticle}
\usepackage{verbatim}
\usepackage[usenames]{color}
\usepackage{amssymb,amsfonts,amsmath,enumerate,latexsym,yhmath,amsthm}
\usepackage{url}
\usepackage{graphicx}

\usepackage{tikz}
\usetikzlibrary{snakes,shapes}
\tikzstyle{tre}=[circle,draw,minimum size=3mm]
\tikzstyle{btre}=[circle,draw,minimum size=4.5mm]
\newcommand{\etq}[1]{%
\draw (#1) node {\tiny $#1$};}

\renewcommand{\leq}{\leqslant}
\renewcommand{\geq}{\geqslant}
\newcommand{\NN}{\mathbb{N}}

\newcommand{\TT}{\mathcal{T}}
\newcommand{\TB}{\mathcal{BT}}
\newcommand{\BT}{\mathcal{BT}}
\newcommand{\FF}{\mathcal{F}}
\newcommand{\FB}{\mathcal{BF}}

\newcommand{\RR}{\mathbb{R}}

 \theoremstyle{plain}
 \newtheorem{theorem}{Theorem}
 \newtheorem{lemma}[theorem]{Lemma}
 \newtheorem{corollary}[theorem]{Corollary}
 \newtheorem{proposition}[theorem]{Proposition}
 \theoremstyle{definition}
 
 \newtheorem{definition}[theorem]{Definition}
 \newtheorem{remark}[theorem]{Remark}

\begin{document}
\begin{frontmatter}

\title{A new balance index for phylogenetic trees}

\author[uib1]{Arnau Mir}
\ead{arnau.mir@uib.es}
\author[uib1]{Francesc Rossell\'o\corref{cor1}}
\ead{cesc.rossello@uib.es}
\author[uib2]{Luc\'\i a Rotger}
\ead{luciarotger87@gmail.com}
\cortext[cor1]{Corresponding author}
\address[uib1]{Research Institute of Health Science (IUNICS) and Department of Mathematics and  Computer Science,
  University of the Balearic Islands,
  E-07122 Palma de
  Mallorca, Spain}
\address[uib2]{Department of Mathematics and  Computer Science,
  University of the Balearic Islands,
  E-07122 Palma de
  Mallorca, Spain}

\begin{abstract}
Several indices that measure the degree of balance of a rooted phylogenetic tree have been proposed so far in the literature. In this work we define and study a new index of this kind, which we call the \emph{total cophenetic index}: the sum, over all pairs of different leaves, of the depth of their least common ancestor. This index makes sense for arbitrary trees, can be computed in linear time and it has a larger range of values and a greater resolution power than other indices like Colless' or Sackin's. We compute its maximum and minimum values for arbitrary and binary trees, as well as exact formulas for its expected value  for binary trees under the Yule and the uniform models of evolution. As a byproduct of this study, we obtain an exact formula for the expected value of the Sackin index under the uniform model, a result that seems to be new in the literature.
\end{abstract}
\end{frontmatter}

\section{Introduction}

A \emph{phylogenetic tree} is a representation of the shared evolutionary history of a set of extant species. From the mathematical point a view, it is a leaf-labeled rooted tree, with its leaves representing  the extant species under study, its internal nodes representing common ancestors of some of them, the root representing the most recent common ancestor of all of them, and the arcs representing direct descendants through mutations. 

One of the most thoroughly studied shape properties of phylogenetic trees is their balance, that is, the degree to which the children of internal nodes tend to have the same number of descendant taxa.  This global degree of balance of a tree is usually quantified by means of a single number generically called an \emph{balance  index}.  The  two most popular  balance indices are  Sackin's \cite{Sackin:72} and Colless' \cite{Colless:82} (see \S 2.2), but there are many more
 \cite[Chap. 33]{fel:04}, and Shao and Sokal \cite[p. 1990]{Shao:90} explicitly advise to use more than one such index to quantify tree balance.

Such balance indices only depend on the topology of the trees, not on the branch lengths or the actual taxa labeling their leaves. Since it is believed that the raw topology of a phylogenetic tree already reflects, at least to some extent, the evolutionary processes that have produced it  \cite[Chap. 33]{fel:04}, these indices have also been widely used  as tools to test stochastic models  of evolution \cite{Mooers97,Shao:90}. 

Two of the most popular stochastic models of evolutionary tree growth are the Yule and the uniform models. The \emph{Yule}, or  \emph{Equal-Rate Markov} model \cite{Harding71,Yule}, starts with a single node and, at every step, a leaf is chosen randomly and uniformly, and it is replaced by a \emph{cherry}, i.e., a phylogenetic tree consisting only of a root and two leaves. Finally, once the desired number of leaves is reached, the labels are assigned randomly and uniformly to the leaves. This corresponds to a model  of evolution where, at each step, each currently extant species can give rise  with the same probability to two new species. Under this model different trees with the same number of leaves may have different probabilities. In contrast, the main feature of the \emph{uniform}, or \emph{Proportional to Distinguishable Arrangements}  model \cite{Rosen78} is that all phylogenetic trees with the same number of leaves have the same probability. From the point of view of tree growth \cite{CS,cherries}, this corresponds to a process where, starting with a node labeled 1, at  the $k$-th step a new pendant arc, ending in the leaf labeled $k+1$, is added either to a new root or to some edge (being all possible locations of this new pendant arc equiprobable). Notice that this is not an explicit model of evolution, only of tree growth. Several properties of the distributions of Sackin's  and Colless'  indices have been studied in the literature under these models \cite{BF:05,BFJ:06,Heard92,KiSl:93,Mul11,Rogers:93,Rogers:94,Rogers:96,SM01}.

In this paper we propose a new balance index, the \emph{total cophenetic index}. It is defined as the sum of the cophenetic values \cite{Sokal:62} of all  pairs of different leaves.  The main features of our index are that, unlike Colless' index, it makes sense for arbitrary (i.e., not necessarily fully resolved) trees; as Colless' and Sackin's indices, it can be easily computed in linear time; its range of values is  larger than Colless' and Sackin's (up to $O(n^3)$, instead of $O(n^2)$), and it has a greater resolution power than those indices. 

We compute the maximum and minimum values of our index, both in the arbitrary and the binary cases, and explicit formulas for its average value under the Yule and the uniform models for binary trees. We actually deduce its average value under the uniform model from an explicit formula for the average value of the Sackin index. This average value was known until now only for its limit distribution \cite{BFJ:06}, and our formula seems thus to be new in the literature.

The rest of this paper is organized as follows. In a first section we introduce the basic notations and facts on phylogenetic trees that will be used henceforth, and we recall some basic facts on the Sackin and the Colless indices. Then, in Section 3,  we define our total cophenetic index $\Phi$ and we establish its basic properties. In Section 4 we compute its maximum and minimum values, and then, in subsequent sections, we compute its expected value under the Yule and the uniform models. We finally devote a last section to conclusions and the discussion of two preliminary numerical experiments involving $\Phi$.

\section{Preliminaries}
\label{sec:prel}

\subsection{Phylogenetic trees}

In this paper, by a  \emph{phylogenetic tree} on a set $S$ of taxa we mean a  rooted tree  with its leaves bijectively labeled in  the set $S$.  To simplify the language, we shall always identify a leaf of a phylogenetic tree with its label.  We shall use the term \emph{phylogenetic tree with $n$ leaves} to refer to a phylogenetic tree on the set $\{1,\ldots,n\}$.  We shall denote by $L(T)$ the set of leaves of a phylogenetic tree $T$ and by $V_{int}(T)$ its set of internal nodes.

A phylogenetic tree is \emph{binary}, or \emph{fully resolved}, when all its internal nodes are \emph{bifurcating}, that is, when every internal node has exactly two children.

Whenever there exists a path from $u$ to $v$ in a phylogenetic tree $T$, we shall say that $v$ is a  \emph{descendant} of $u$ and also that $u$ is an  \emph{ancestor} of $v$.  The \emph{cluster} of a node $v$ in $T$ is the set $C_T(v)$ of its descendant leaves, an we shall denote by $\kappa_T(v)$ the cardinal $|C_T(v)|$, that is, the number of descendant leaves of $v$.

Given a node $v$ of a phylogenetic tree $T$, the \emph{subtree of $T$ rooted at $v$} is the subgraph of $T$ induced on the set of descendants of $v$. It is a phylogenetic tree on $C_T(v)$ with root this node $v$.

The \emph{lowest common ancestor} (LCA) of a pair of nodes $u,v$ of a phylogenetic tree $T$, in symbols $LCA_T(u,v)$, is the unique common ancestor of them that is a descendant of every other common ancestor of them.  

The \emph{depth} $\delta_T(v)$ of a node $v$  in a phylogenetic tree $T$ is the length (in number of arcs) of the unique path from the root $r$ to $v$. 

A \emph{rooted caterpillar} is a binary phylogenetic tree  all whose internal nodes have a leaf child: see Fig. \ref{fig:exs}.(a). A \emph{rooted  star} is a phylogenetic tree   such that all its leaves have depth 1: see Fig. \ref{fig:exs}.(b).

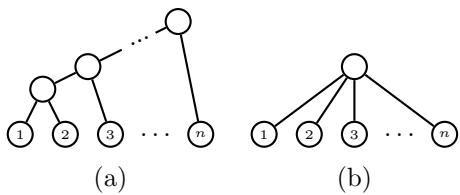
\begin{figure}[htb]
\begin{center}
\begin{tikzpicture}[thick,>=stealth,scale=0.3]
\draw(0,0) node [tre] (1) {};  \etq 1
\draw(2,0) node [tre] (2) {};  \etq 2
\draw(4,0) node [tre] (3) {};  \etq 3
\draw(6,0) node  {$\ldots$};  
\draw(8,0) node [tre] (n) {};  \etq n
\draw(1,2) node[tre] (a) {};
\draw(3,3) node[tre] (b) {};
\draw(7,5) node[tre] (r) {};
\draw(5,4) node  {.};
\draw(5.3,4.15) node  {.};
\draw(5.6,4.3) node  {.};
\draw (a)--(1);
\draw (a)--(2);
\draw (b)--(3);
\draw (b)--(a);
\draw (b)-- (4.5,3.75);
\draw (r)-- (6,4.5);
\draw (r)-- (n);
\draw(4,-2) node {(a)};
\end{tikzpicture}
\quad
\begin{tikzpicture}[thick,>=stealth,scale=0.3]
\draw(0,0) node [tre] (1) {};  \etq 1
\draw(2,0) node [tre] (2) {};  \etq 2
\draw(4,0) node [tre] (3) {};  \etq 3
\draw(6,0) node  {$\ldots$};  
\draw(8,0) node [tre] (n) {};  \etq n
\draw(4,3) node[tre] (r) {};
\draw  (r)--(1);
\draw  (r)--(2);
\draw  (r)--(3);
\draw  (r)--(n);
\draw(4,-2) node {(b)};
\end{tikzpicture}
\end{center}
\caption{\label{fig:exs} 
(a) A rooted caterpillar with $n$ leaves. (b) The rooted star  with $n$ leaves.}
\end{figure}

Let $T$ be a binary phylogenetic tree. For every $v\in V_{int}(T)$, say with children $v_1,v_2$,  the \emph{balance value} of $v$ is $bal_T(v)=|\kappa_T(v_1)-\kappa_T(v_2)|$. An internal node $v$ of  $T$ is \emph{balanced} when $bal_T(v)\leq 1$. So, a node $v$  with children $v_1$ and $v_2$ is  balanced if, and only if, $\{\kappa_T(v_1),\kappa_T(v_2)\}=\{\lfloor \kappa_T(v)/2\rfloor,\lceil \kappa_T(v)/2\rceil\}$.

We shall say that a binary phylogenetic tree $T$ is \emph{maximally balanced} when all its internal nodes are balanced. Recurrently, a binary phylogenetic tree is maximally balanced when its root is balanced and both subtrees rooted at the children of the root are maximally balanced. Notice that, for any number $n$ of nodes, the topology of a maximally balanced tree with $n$ leaves is fixed, and therefore two maximally balanced trees with the same number of leaves differ only in their labeling. Fig. \ref{fig:bal} depicts the maximally balanced trees with $n=2,\ldots,6$ leaves, up to relabelings.

 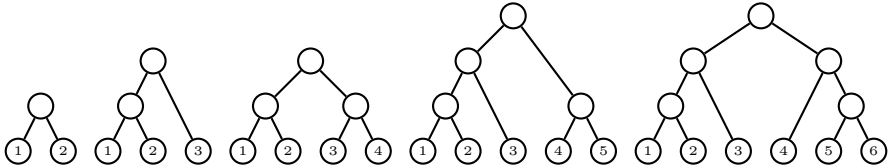
\begin{figure}[htb]
\begin{center}
\begin{tikzpicture}[thick,>=stealth,scale=0.3]
\draw(0,0) node [tre] (1) {};  \etq 1
\draw(2,0) node [tre] (2) {};  \etq 2
\draw(1,2) node[tre] (a) {};
\draw  (a)--(1);
\draw  (a)--(2);
\end{tikzpicture}
\ 
\begin{tikzpicture}[thick,>=stealth,scale=0.3]
\draw(0,0) node [tre] (1) {};  \etq 1
\draw(2,0) node [tre] (2) {};  \etq 2
\draw(4,0) node [tre] (3) {};  \etq 3
\draw(1,2) node[tre] (a) {};
\draw(2,4) node[tre] (b) {};
\draw  (a)--(1);
\draw  (a)--(2);
\draw  (b)--(a);
\draw  (b)--(3);
\end{tikzpicture}
\ 
\begin{tikzpicture}[thick,>=stealth,scale=0.3]
\draw(0,0) node [tre] (1) {};  \etq 1
\draw(2,0) node [tre] (2) {};  \etq 2
\draw(4,0) node [tre] (3) {};  \etq 3
\draw(6,0) node [tre] (4) {};  \etq 4
\draw(1,2) node[tre] (a) {};
\draw(5,2) node[tre] (c) {};
\draw(3,4) node[tre] (b) {};
\draw  (a)--(1);
\draw  (a)--(2);
\draw  (b)--(a);
\draw  (b)--(c);
\draw  (c)--(3);
\draw  (c)--(4);
\end{tikzpicture}
\ 
\begin{tikzpicture}[thick,>=stealth,scale=0.3]
\draw(0,0) node [tre] (1) {};  \etq 1
\draw(2,0) node [tre] (2) {};  \etq 2
\draw(4,0) node [tre] (3) {};  \etq 3
\draw(6,0) node [tre] (4) {};  \etq 4
\draw(8,0) node [tre] (5) {};  \etq 5
\draw(1,2) node[tre] (a) {};
\draw(7,2) node[tre] (c) {};
\draw(2,4) node[tre] (b) {};
\draw(4,6) node[tre] (d) {};
\draw  (a)--(1);
\draw  (a)--(2);
\draw  (b)--(a);
\draw  (b)--(3);
\draw  (c)--(5);
\draw  (c)--(4);
\draw  (d)--(b);
\draw  (d)--(c);
\end{tikzpicture}
\ 
\begin{tikzpicture}[thick,>=stealth,scale=0.3]
\draw(0,0) node [tre] (1) {};  \etq 1
\draw(2,0) node [tre] (2) {};  \etq 2
\draw(4,0) node [tre] (3) {};  \etq 3
\draw(6,0) node [tre] (4) {};  \etq 4
\draw(8,0) node [tre] (5) {};  \etq 5
\draw(10,0) node [tre] (6) {};  \etq 6
\draw(1,2) node[tre] (a) {};
\draw(9,2) node[tre] (c) {};
\draw(2,4) node[tre] (b) {};
\draw(8,4) node[tre] (d) {};
\draw(5,6) node[tre] (r) {};
\draw  (a)--(1);
\draw  (a)--(2);
\draw  (b)--(a);
\draw  (b)--(3);
\draw  (c)--(5);
\draw  (c)--(6);
\draw  (d)--(4);
\draw  (d)--(c);
\draw  (r)--(b);
\draw  (r)--(d);
\end{tikzpicture}
\end{center}
\caption{\label{fig:bal} 
Maximally balanced trees.}
\end{figure}

Let $\TT_n$ (resp., $\TB_n$) be the set of isomorphism classes of phylogenetic trees (resp, binary phylogenetic trees)  with $n$ leaves.  It is well known \cite[Ch. 3]{fel:04} that $|\TB_1|=1$ and, for every $n\geq 2$,
$$
|\TB_n|=(2n-3)!!=(2n-3)(2n-5)\cdots 3 \cdot 1.
$$
No  closed formula is known for the cardinal $|\TT_n|$, only recurrences or generating functions (see again \cite[Ch. 3]{fel:04} and the references therein).

An \emph{ordered $m$-forest} on a set $S$  is an ordered sequence of $m$ phylogenetic trees $(T_1,T_2,\ldots,T_m)$, each $T_i$ on a set $S_i$ of taxa, such that these sets $S_i$ are pairwise disjoint and their union is $S$.  An ordered forest is \emph{binary} when it consists of binary trees.
Let $\FF_{m,n}$ (resp., $\FB_{m,n}$) be the set of isomorphism classes of ordered $m$-forests 
(resp., binary ordered $m$-forests)  on a set $S$ with $|S|=n$. It is known (see, for instance, \cite[Lem. 1]{MirR10}) that for every $n\geq m\geq 1$,
$$
|\FB_{m,n}|= \frac{(2n-m-1)!m}{(n-m)!2^{n-m}}.
$$
Again, no closed formula is known for $|\FF_{m,n}|$.

\subsection{Balance indices}

Several balance indices have been proposed so far in the literature \cite[p. 563]{fel:04}. The two most popular ones are  the \emph{Sackin index} \cite{Sackin:72} and the \emph{Colless index} \cite{Colless:82}.  The Sackin index of a phylogenetic tree $T\in \TT_n$ is defined as the sum of the depths of its leaves:
$$
S(T)=\sum_{i=1}^n\delta_T(i).
$$
Alternatively \cite{BF:05},
$$
S(T)=\sum_{v\in V_{int}(T)} \kappa_T(v).
$$
On the other hand, the Colless index of a \emph{binary} phylogenetic tree  $T$ is defined as
$$
C(T)=\sum_{v\in V_{int}(T)} bal_T(v).
$$
This Colless index has been extended to non-binary trees by defining $bal_T(v)=0$ for every non-bifurcating internal node \cite{Shao:90}.

It is straightforward to notice that these two indices depend only on the topology of the tree, and they are invariant under isomorphisms and relabelings of leaves. This is desirable in a balance index, because the  degree of symmetry of a tree depends only on its shape.

Both Sackin's and Colless's indices reach their maximum value exactly at caterpillars, which are clearly the more imbalanced trees, and they reach their minimum on $\TB_n$ at the maximally balanced trees \cite{Heard92,Shao:90}. In both cases, the maximum value is in $O(n^2)$.
But they may also reach their minimum on $\TB_n$ at other trees.  For instance, for $n=6$, both indices take their minimum value at the two trees $T,T'$ depicted in Fig. \ref{fig:minind}. 
$T'$ is maximally balanced, but $T$ is not so. Actually, it is easy to check that Sackin's index is invariant under interchanges of cousins, which may produce trees with different degrees of symmetry but the same Sackin index. 

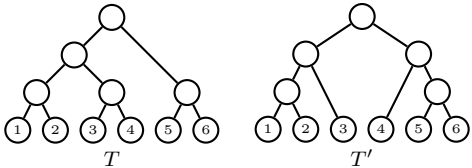
\begin{figure}[htb]
\begin{center}
\begin{tikzpicture}[thick,>=stealth,scale=0.25]
\draw(0,0) node[tre] (1) {};   \etq 1
\draw(2,0) node[tre] (2) {};   \etq 2
\draw(4,0) node[tre] (3) {};   \etq 3
\draw(6,0) node[tre] (4) {};   \etq 4
\draw(8,0) node[tre] (5) {};   \etq 5
\draw(10,0) node[tre] (6) {};   \etq 6
\draw(1,2) node[tre] (a) {};   
\draw(5,2) node[tre] (b) {};   
\draw(9,2) node[tre] (c) {};   
\draw(3,4) node[tre] (d) {};   
\draw(5,6) node[tre] (r) {};   
\draw (r)--(d);
\draw (r)--(c);
\draw (d)--(a);
\draw (d)--(b);
\draw (a)--(1);
\draw (a)--(2);
\draw (b)--(3);
\draw (b)--(4);
\draw (c)--(5);
\draw (c)--(6);
\draw(5,-1.5) node  {\footnotesize $T$};
\end{tikzpicture}
\quad
\begin{tikzpicture}[thick,>=stealth,scale=0.25]
\draw(0,0) node[tre] (1) {};   \etq 1
\draw(2,0) node[tre] (2) {};   \etq 2
\draw(4,0) node[tre] (3) {};   \etq 3
\draw(6,0) node[tre] (4) {};   \etq 4
\draw(8,0) node[tre] (5) {};   \etq 5
\draw(10,0) node[tre] (6) {};   \etq 6
\draw(1,2) node[tre] (a) {};   
\draw(2,4) node[tre] (b) {};   
\draw(9,2) node[tre] (c) {};   
\draw(8,4) node[tre] (d) {};   
\draw(5,6) node[tre] (r) {};   
\draw (r)--(d);
\draw (r)--(b);
\draw (b)--(a);
\draw (b)--(3);
\draw (a)--(1);
\draw (a)--(2);
\draw (d)--(c);
\draw (d)--(4);
\draw (c)--(5);
\draw (c)--(6);
\draw(5,-1.5) node  {\footnotesize $T'$};
\end{tikzpicture}
\end{center}
\caption{\label{fig:minind} 
Two trees having minimum Sackin's and Colless' indices on $\TB_6$: $S(T)=S(T')=16$ and $C(T)=C(T')=2$.}
\end{figure}

The main drawback with Colless' index is its difficult meaningful generalization to non-binary trees. Moreover, as Fig. \ref{fig:minind} shows, although not every interchange of cousins yields trees with the same Colless index, there are still interchanges of cousins that modify  the symmetry of the trees but  preserve this index.

The expected values  of these indices on $\BT_n$ have been studied under the Yule and the uniform models. Recall that, under the Yule model, different trees in $\BT_n$ may have different probabilities: namely, a tree $T$ with $n$ leaves has probability  \cite{Brown,SM01}
$$
P_Y(T)=\frac{2^{n-1}}{n!}\prod_{v\in V_{int}(T)}\frac{1}{\kappa_T(v)-1}.
$$
Under the uniform model, all trees in $\BT_n$ are equiprobable, and thus they have probability
$$
P_U(T)=\frac{1}{(2n-3)!!}.
$$

Let $C_n$ and $S_n$ be the random variables defined by choosing a tree $T\in \BT_n$ and computing $C(T)$ or $S(T)$, respectively. The following facts are known about the expected values of these random variables:
\begin{itemize}
\item Under the Yule model,
\begin{itemize}
\item  $E_Y(C_n)=(n\hspace*{-1ex} \mod 2)+n\sum\limits _{j=2}^{\lfloor n/2\rfloor} 1/j$  \cite{Heard92}.

\item  $E_Y(S_n)=2n\sum\limits_{j=2}^{n}1/j$ \cite{KiSl:93}.
\end{itemize}

\item Under the uniform model,
$$
E_U(C_n), E_U(S_n)\sim \sqrt{\pi}n^{3/2}\quad \mbox{\cite{BFJ:06}}.
$$
We shall actually prove in this paper (see Theorem \ref{th:sackin}) that 
$$
E_U(S_n)=\frac{n }{2n-3}\, {}_3 F_2
\bigg(\begin{array}{l} 2,\ 2,\ 2-n \\[-0.5ex] 1,\  4-2n\end{array};2\bigg),
$$
 where
${}_3 F_2$ is a \emph{hypergeometric function}  \cite{Bayley}.
\end{itemize}

\section{The total cophenetic index}
For every pair of leaves $i,j$ in a phylogenetic tree $T$, their \emph{cophenetic value} \cite{Sokal:62} is the depth of their least common ancestor:
$$
\varphi_T(i,j)=\delta_T(LCA_T(i,j)).
$$ 

\begin{definition}
The \emph{total cophenetic index} of a phylogenetic tree $T\in\TT_n$ is the sum of the cophenetic values of its pairs of different leaves:
$$
 \Phi(T)=\sum_{1\leq i<j\leq n} \varphi_T(i,j).
 $$
\end{definition}

This index can be seen as an extension of Sackin's: instead of adding up the depths of the leaves (that is, the depths of the LCA of every leaf and itself), $\Phi(T)$ adds up the depths of the LCA of every pair of leaves in $T$.  Notice also that, as Sackin's and Colless' indices, $\Phi(T)$ only depends on the topology of $T$, and in particular it is invariant under permutations of its labels.

Fig. \ref{fig:ex1} shows all possible topologies of phylogenetic trees with 5 leaves, and their  total cophenetic indices. Although we shall return on it later for trees with an arbitrary number $n$ of leaves, notice that the rooted star has the smallest total cophenetic value, 0;  the binary tree with the smallest total cophenetic value is the maximally balanced;  and the tree with the largest  total cophenetic value is the caterpillar.

\begin{figure}[htb]
\begin{center}
\begin{tikzpicture}[thick,>=stealth,scale=0.25]
\draw(0,0) node[tre] (1) {}; \etq 1
\draw(2,0) node[tre] (2) {}; \etq 2
\draw(4,0) node[tre] (3) {}; \etq 3
\draw(6,0) node[tre] (4) {}; \etq 4
\draw(8,0) node[tre] (5) {}; \etq 5
\draw(4,5) node[tre] (r) {}; 
\draw (r)--(1);
\draw (r)--(2);
\draw (r)--(3);
\draw (r)--(4);
\draw (r)--(5);
\draw(3,-2) node {\footnotesize $\Phi(T)=0$};
\end{tikzpicture}
\
\begin{tikzpicture}[thick,>=stealth,scale=0.25]
\draw(0,0) node[tre] (1) {}; \etq 1
\draw(2,0) node[tre] (2) {}; \etq 2
\draw(4,0) node[tre] (3) {}; \etq 3
\draw(6,0) node[tre] (4) {}; \etq 4
\draw(8,0) node[tre] (5) {}; \etq 5
\draw(1,3) node[tre] (z1) {}; 
\draw(3,5) node[tre] (r) {}; 
\draw (z1)--(1);
\draw (z1)--(2);
\draw (r)--(3);
\draw (r)--(4);
\draw (r)--(z1);
\draw (r)--(5);
\draw(3,-2) node {\footnotesize $\Phi(T)=1$};
\end{tikzpicture}
\
\begin{tikzpicture}[thick,>=stealth,scale=0.25]
\draw(0,0) node[tre] (1) {}; \etq 1
\draw(2,0) node[tre] (2) {}; \etq 2
\draw(4,0) node[tre] (3) {}; \etq 3
\draw(6,0) node[tre] (4) {}; \etq 4
\draw(8,0) node[tre] (5) {}; \etq 5
\draw(1,2) node[tre] (z1) {}; 
\draw(5,2) node[tre] (z2) {}; 
\draw(4.5,5) node[tre] (r) {}; 
\draw (z1)--(1);
\draw (z1)--(2);
\draw (z2)--(3);
\draw (z2)--(4);
\draw (r)--(z1);
\draw (r)--(z2);
\draw (r)--(5);
\draw(3,-2) node {\footnotesize $\Phi(T)=2$};
\end{tikzpicture}
\
\begin{tikzpicture}[thick,>=stealth,scale=0.25]
\draw(0,0) node[tre] (1) {}; \etq 1
\draw(2,0) node[tre] (2) {}; \etq 2
\draw(4,0) node[tre] (3) {}; \etq 3
\draw(6,0) node[tre] (4) {}; \etq 4
\draw(8,0) node[tre] (5) {}; \etq 5
\draw(2,2) node[tre] (z1) {}; 
\draw(4.5,5) node[tre] (r) {}; 
\draw (z1)--(1);
\draw (z1)--(2);
\draw (z1)--(3);
\draw (r)--(4);
\draw (r)--(z1);
\draw (r)--(5);
\draw(3,-2) node {\footnotesize $\Phi(T)=3$};
\end{tikzpicture}
\
\begin{tikzpicture}[thick,>=stealth,scale=0.25]
\draw(0,0) node[tre] (1) {}; \etq 1
\draw(2,0) node[tre] (2) {}; \etq 2
\draw(4,0) node[tre] (3) {}; \etq 3
\draw(6,0) node[tre] (4) {}; \etq 4
\draw(8,0) node[tre] (5) {}; \etq 5
\draw(1,2) node[tre] (z1) {}; 
\draw(2.5,4) node[tre] (z2) {}; 
\draw(5,6) node[tre] (r) {}; 
\draw (z1)--(1);
\draw (z1)--(2);
\draw (z2)--(3);
\draw (r)--(4);
\draw (r)--(z2);
\draw (z2)--(z1);
\draw (r)--(5);
\draw(3,-2) node {\footnotesize $\Phi(T)=4$};
\end{tikzpicture}
\
\begin{tikzpicture}[thick,>=stealth,scale=0.25]
\draw(0,0) node[tre] (1) {}; \etq 1
\draw(2,0) node[tre] (2) {}; \etq 2
\draw(4,0) node[tre] (3) {}; \etq 3
\draw(6,0) node[tre] (4) {}; \etq 4
\draw(8,0) node[tre] (5) {}; \etq 5
\draw(2,3) node[tre] (z1) {}; 
\draw(4.5,6) node[tre] (r) {}; 
\draw(7,3) node[tre] (z2) {}; 
\draw (z1)--(1);
\draw (z1)--(2);
\draw (z1)--(3);
\draw (z2)--(4);
\draw (r)--(z1);
\draw (r)--(z2);
\draw (z2)--(5);
\draw(3,-2) node {\footnotesize $\Phi(T)=4$};
\end{tikzpicture}
\
\begin{tikzpicture}[thick,>=stealth,scale=0.25]
\draw(0,0) node[tre] (1) {}; \etq 1
\draw(2,0) node[tre] (2) {}; \etq 2
\draw(4,0) node[tre] (3) {}; \etq 3
\draw(6,0) node[tre] (4) {}; \etq 4
\draw(8,0) node[tre] (5) {}; \etq 5
\draw(1,2) node[tre] (z1) {}; 
\draw(2.5,4) node[tre] (z2) {}; 
\draw(5,6) node[tre] (r) {}; 
\draw(7,3) node[tre] (z3) {}; 
\draw (z1)--(1);
\draw (z1)--(2);
\draw (z2)--(3);
\draw (z3)--(4);
\draw (r)--(z3);
\draw (z2)--(z1);
\draw (r)--(z2);
\draw (z3)--(5);
\draw(3,-2) node {\footnotesize $\Phi(T)=5$};
\end{tikzpicture}
\
\begin{tikzpicture}[thick,>=stealth,scale=0.25]
\draw(0,0) node[tre] (1) {}; \etq 1
\draw(2,0) node[tre] (2) {}; \etq 2
\draw(4,0) node[tre] (3) {}; \etq 3
\draw(6,0) node[tre] (4) {}; \etq 4
\draw(8,0) node[tre] (5) {}; \etq 5
\draw(4,6) node[tre] (r) {}; 
\draw(3,3) node[tre] (z) {}; 
\draw (z)--(1);
\draw (z)--(2);
\draw (z)--(3);
\draw (z)--(4);
\draw (r)--(z);
\draw (r)--(5);

\draw(3,-2) node {\footnotesize $\Phi(T)=6$};
\end{tikzpicture}
\
\begin{tikzpicture}[thick,>=stealth,scale=0.25]
\draw(0,0) node[tre] (1) {}; \etq 1
\draw(2,0) node[tre] (2) {}; \etq 2
\draw(4,0) node[tre] (3) {}; \etq 3
\draw(6,0) node[tre] (4) {}; \etq 4
\draw(1,2) node[tre] (z) {}; 
\draw(2.5,4) node[tre] (r) {}; 
\draw(8,0) node[tre] (5) {}; \etq 5
\draw(4,6) node[tre] (r0) {}; 
\draw (z)--(1);
\draw (z)--(2);
\draw (r)--(3);
\draw (r)--(4);
\draw (r)--(z);
\draw (r0)--(r);
\draw (r0)--(5);
\draw(3,-2) node {\footnotesize $\Phi(T)=7$};
\end{tikzpicture}
\
\begin{tikzpicture}[thick,>=stealth,scale=0.25]
\draw(0,0) node[tre] (1) {}; \etq 1
\draw(2,0) node[tre] (2) {}; \etq 2
\draw(4,0) node[tre] (3) {}; \etq 3
\draw(6,0) node[tre] (4) {}; \etq 4
\draw(1,2) node[tre] (z) {}; 
\draw(5,2) node[tre] (x) {}; 
\draw(2.5,4) node[tre] (r) {}; 
\draw(8,0) node[tre] (5) {}; \etq 5
\draw(4,6) node[tre] (r0) {}; 
\draw (z)--(1);
\draw (z)--(2);
\draw (x)--(4);
\draw (x)--(3);
\draw (r)--(z);
\draw (r)--(x);
\draw (r0)--(r);
\draw (r0)--(5);
\draw(3,-2) node {\footnotesize $\Phi(T)=8$};
\end{tikzpicture}
\
\begin{tikzpicture}[thick,>=stealth,scale=0.25]
\draw(0,0) node[tre] (1) {}; \etq 1
\draw(2,0) node[tre] (2) {}; \etq 2
\draw(4,0) node[tre] (3) {}; \etq 3
\draw(6,0) node[tre] (4) {}; \etq 4
\draw(2,2) node[tre] (z) {}; 
\draw(3,4) node[tre] (r) {}; 
\draw(8,0) node[tre] (5) {}; \etq 5
\draw(4,6) node[tre] (r0) {}; 
\draw (z)--(1);
\draw (z)--(2);
\draw (z)--(3);
\draw (r)--(4);
\draw (r)--(z);
\draw (r0)--(r);
\draw (r0)--(5);
\draw(3,-2) node {\footnotesize $\Phi(T)=9$};
\end{tikzpicture}
\
\begin{tikzpicture}[thick,>=stealth,scale=0.25]
\draw(0,0) node[tre] (1) {}; \etq 1
\draw(2,0) node[tre] (2) {}; \etq 2
\draw(4,0) node[tre] (3) {}; \etq 3
\draw(6,0) node[tre] (4) {}; \etq 4
\draw(1,2) node[tre] (z) {}; 
\draw(2.5,3.33) node[tre] (x) {}; 
\draw(4,4.66) node[tre] (r) {}; 
\draw(8,0) node[tre] (5) {}; \etq 5
\draw(5.5,6) node[tre] (r0) {}; 
\draw (z)--(1);
\draw (z)--(2);
\draw (x)--(z);
\draw (x)--(3);
\draw (r)--(4);
\draw (r)--(x);
\draw (r0)--(r);
\draw (r0)--(5);
\draw(3,-2) node {\footnotesize $\Phi(T)=10$};
\end{tikzpicture}
\end{center}
\caption{\label{fig:ex1} All phylogenetic trees with 5 leaves, up to relabelings, and their total cophenetic index.}
\end{figure}
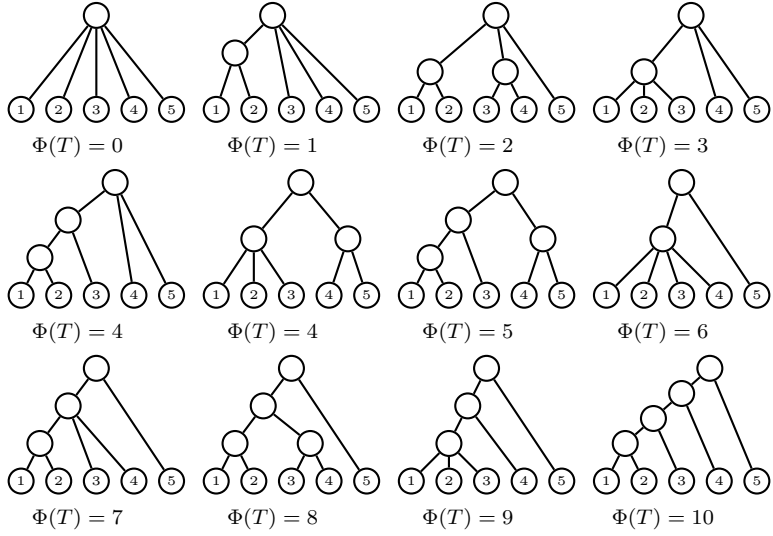

The following alternative expression for $\Phi(T)$ will be useful in many proofs.

\begin{lemma}\label{lem:alt1}
Let $T\in \TT_n$ be a phylogenetic tree with root $r$. Then,
$$
\Phi(T)=\sum_{v\in V_{int}(T)-\{r\}} \binom{\kappa_T(v)}{2}.
$$
\end{lemma}

\begin{proof}
For every $v\in V_{int}(T)-\{r\}$ and for every $i,j\in L(T)$, let
$$
\gamma_{v}(i,j)=\left\{\begin{array}{ll}
1 & \mbox{ if $i,j\in C_T(v)$}\\
0 & \mbox{ otherwise}
\end{array}\right.
$$
Then, $\phi_T(i,j)=\sum\limits_{v\in V_{int}(T)-\{r\}} \gamma_{v}(i,j)$ and thus
$$
\begin{array}{rl}
\Phi(T)& \displaystyle =\sum_{1\leq i<j\leq n}\sum_{V_{int}(T)-\{r\}} \gamma_{v}(i,j)
=\sum_{v\in V_{int}(T)-\{r\}}\sum_{1\leq i<j\leq n} \gamma_{v}(i,j)\\[1ex] & \displaystyle
=\sum_{v\in V_{int}(T)-\{r\}}\binom{|C_T(v)|}{2}. 
\end{array}
$$
\end{proof}

\begin{corollary}
For every $T\in \TT_n$, $\Phi(T)$ can be computed in time $O(n)$.
\end{corollary}

\begin{proof}
The vector $(\kappa_T(v))_{v\in V_{int}(T)-\{r\}}$ can be computed in linear time by traversing in post order the tree $T$ \cite[\S 3.2]{Val:02}, and then, by the last lemma, $\Phi(T)$ is computed in linear time from this vector.
\end{proof}

\begin{lemma}\label{lem:rec}
Let $T\in \TT_n$ be a phylogenetic tree with root $r$, and let $T_1,\ldots,T_k$, $k\geq 2$, be the subtrees rooted at the children of $r$; cf. Fig \ref{fig:rec}. Then,
$$
\Phi(T)=\sum_{i=1}^k\Phi(T_i)+\sum_{i=1}^k\binom{|L(T_i)|}{2}.
$$
\end{lemma}

\begin{figure}[htb]
\begin{center}
\begin{tikzpicture}[thick,>=stealth,scale=0.25]
\draw(0,0) node[tre] (z1) {}; 
\draw (z1)--(-2,-3)--(2,-3)--(z1);
\draw(0,-2) node  {\footnotesize $T_1$};
\draw(5,0) node[tre] (z2) {}; 
\draw (z2)--(3,-3)--(7,-3)--(z2);
\draw(5,-2) node  {\footnotesize $T_2$};
\draw(8,-2.8) node {.}; 
\draw(8.4,-2.8) node {.}; 
\draw(8.8,-2.8) node {.}; 
\draw(11.5,0) node[tre] (z3) {}; 
\draw (z3)--(9.5,-3)--(14.5,-3)--(z3);
\draw(11.5,-2) node  {\footnotesize $T_k$};
\draw(7.25,4) node[tre] (z) {}; 
\draw (z)--(z1);
\draw (z)--(z2);
\draw (z)--(z3);
\end{tikzpicture}
\end{center}
\caption{\label{fig:rec}}
\end{figure}

\begin{proof}
Let $z_i$ be the root of $T_i$, $i=1,\ldots,k$, and $r$ the root of $T$. Then, by Lemma \ref{lem:alt1},
$$
\begin{array}{rl}
\Phi(T) & \displaystyle =\sum_{v\in V_{int}(T)-\{r\}} \binom{\kappa_T(v)}{2}
=\sum_{i=1}^{k}\sum_{v\in V_{int}(T_i)} \binom{\kappa_{T_i}(v)}{2}\\ & \displaystyle
=\sum_{i=1}^{k}\Big(\binom{\kappa_{T_i}(z_i)}{2}+\sum_{v\in V_{int}(T_i)-\{z_i\}} \binom{\kappa_{T_i}(v)}{2}\Big) \\ & \displaystyle
=\sum_{i=1}^{k}\Big( \binom{|L(T_i)|}{2}+ \Phi(T_i) \Big).
\end{array}
$$
\end{proof}

This shows that the total cophenetic index is a \emph{recursive tree shape statistic} in the sense of \cite{Matsen}.

Next lemma shows that the total cophenetic index is local, in the sense that if two trees differ only on a rooted subtree, then the difference between their total cophenetic values is equal to that of these subtrees. Sackin's and Colless' indices also satisfy this property.

\begin{lemma}\label{lem:basic}
Let $T_0$ and $T_0'$ be two phylogenetic trees with $L(T_0)=L(T_0')\subseteq \{1,\ldots,n\}$, let $T\in \TT_n$ be such that its subtree rooted at some node $z$ is $T_0$, and let $T'\in\TT_n$ be the tree obtained from $T$ by replacing $T_0$ by $T_0'$ as its subtree rooted at $z$. Then
$$
\Phi(T)-\Phi(T')=\Phi(T_0)-\Phi(T_0').
$$
\end{lemma}

\begin{proof}
Without any loss of generality, assume that $L(T_0)=L(T_0')=\{1,\ldots,m\}$, with $m\leq n$. Let $k=\delta_T(z)=\delta_{T'}(z)$. Then, for every $i,j\leq m$,
$$
\varphi_T(i,j)=k+\varphi_{T_0}(i,j),\ 
\varphi_{T'}(i,j)=k+\varphi_{T'_0}(i,j).
$$
On the other hand, $\varphi_T(i,j)=\varphi_{T'}(i,j)$ if $i>m$ or $j>m$. Therefore
$$
\begin{array}{rl}
\Phi(T)-\Phi(T') & \displaystyle =
\sum_{1\leq i<j\leq m}(\varphi_T(i,j)-\varphi_{T'}(i,j))\\ & \displaystyle =
\sum_{1\leq i<j\leq m}(\varphi_{T_0}(i,j)-\varphi_{T_0'}(i,j))=
\Phi(T_0)-\Phi(T_0').
\end{array}
$$
\end{proof}

The \emph{nodal distance} $d_T(i,j)$ between a pair of leaves $i,j$ is the length of the unique undirected path connecting them; equivalently, it is the sum of the lengths of the paths from $LCA(i,j)$ to $i$ and $j$. The \emph{total area} \cite{MirR10} of a tree $T\in \TT_n$ is defined as
$$
D(T)=\sum_{1\leq i<j\leq n} d_T(i,j).
$$
There is an easy relation between $\Phi(T)$, $S(T)$ and  $D(T)$, which will be used several times in this paper. 

\begin{lemma}\label{lem:total}
For every $T\in \TT_n$,
$$
(n-1)S(T)=2\phi(T)+D(T).
$$
\end{lemma}

\begin{proof}
It is straightforward to check that, for every $i,j\in L(T)$,
$$
\delta_T(i)+\delta_T(j)=d_T(i,j)+2\varphi_T(i,j).
$$
Therefore,
$$
\begin{array}{rl}
2\phi(T)+D(T) & \displaystyle =\sum_{1\leq i<j\leq n}(2\varphi_T(i,j)+d_T(i,j))
=\sum_{1\leq i<j\leq n}(\delta_T(i)+\delta_T(j))\\ & \displaystyle =(n-1)\sum_{i=1}^{n}\delta_T(i)=(n-1)S(T).
\end{array}
$$
\end{proof}

\section{Trees with maximum and minimum $\Phi$}

In this section we determine which  trees in $\TT_n$ and $\TB_n$ have the largest and smallest  total cophenetic indices. We begin by establishing two lemmas that will allow us to find the trees with the maximum $\Phi$ on $\TT_n$.

\begin{lemma}\label{lem:basic1}
Let $T_1, \ldots,T_k$, with $k\geq 3$, be an ordered forest on $\{1,\ldots,m\}$. Consider the trees $T_0,T_0'\in \TT_n$ described in Fig. \ref{fig:1}. Then, $\Phi(T_0')-\Phi(T_0)> 0$.
\end{lemma}

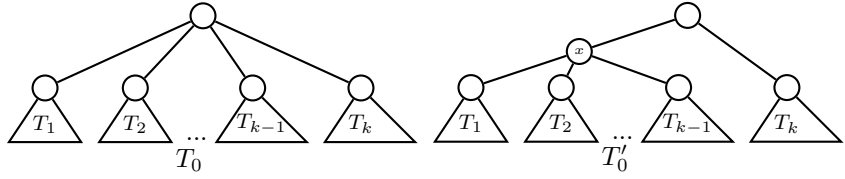
\begin{figure}[htb]
\begin{center}
\begin{tikzpicture}[thick,>=stealth,scale=0.24]
\draw(0,0) node[tre] (z1) {}; 
\draw (z1)--(-2,-3)--(2,-3)--(z1);
\draw(0,-2) node  {\footnotesize $T_1$};
\draw(5,0) node[tre] (z2) {}; 
\draw (z2)--(3,-3)--(7,-3)--(z2);
\draw(5,-2) node  {\footnotesize $T_2$};
\draw(8,-2.8) node {.}; 
\draw(8.4,-2.8) node {.}; 
\draw(8.8,-2.8) node {.}; 
\draw(11.5,0) node[tre] (z3) {}; 
\draw (z3)--(9.5,-3)--(14.5,-3)--(z3);
\draw(12,-2) node  {\footnotesize $T_{k-1}$};
\draw(17.5,0) node[tre] (z4) {}; 
\draw (z4)--(15.5,-3)--(20.5,-3)--(z4);
\draw(17.5,-2) node  {\footnotesize $T_{k}$};
\draw(8.75,4) node[tre] (r) {};  
\draw (r)--(z1);
\draw (r)--(z2);
\draw (r)--(z3);
\draw (r)--(z4);
\draw(8,-4) node  {$T_0$};
\end{tikzpicture}
\
\begin{tikzpicture}[thick,>=stealth,scale=0.24]
\draw(0,0) node[tre] (z1) {}; 
\draw (z1)--(-2,-3)--(2,-3)--(z1);
\draw(0,-2) node  {\footnotesize $T_1$};
\draw(5,0) node[tre] (z2) {}; 
\draw (z2)--(3,-3)--(7,-3)--(z2);
\draw(5,-2) node  {\footnotesize $T_2$};
\draw(8,-2.8) node {.}; 
\draw(8.4,-2.8) node {.}; 
\draw(8.8,-2.8) node {.}; 
\draw(11.5,0) node[tre] (z3) {}; 
\draw (z3)--(9.5,-3)--(14.5,-3)--(z3);
\draw(12,-2) node  {\footnotesize $T_{k-1}$};
\draw(17.5,0) node[tre] (z4) {}; 
\draw (z4)--(15.5,-3)--(20.5,-3)--(z4);
\draw(17.5,-2) node  {\footnotesize $T_{k}$};
\draw(6,2) node[tre] (x) {}; \etq x
\draw(12,4) node[tre] (r) {};  
\draw (x)--(z1);
\draw (x)--(z2);
\draw (x)--(z3);
\draw (r)--(x);
\draw (r)--(z4);
\draw(8,-3.9) node  {$T_0'$};
\end{tikzpicture}

\end{center}
\caption{\label{fig:1} 
The trees $T_0$ and $T_0'$ in the statement of Lemma \ref{lem:basic1}.}
\end{figure}

\begin{proof}
With the notations of Fig. \ref{fig:1}, notice that
$$
\begin{array}{rl}
\Phi(T_0')-\Phi(T_0) & \displaystyle = \sum_{v\in V_{int}(T'_0)-\{r\}} \binom{\kappa_{T_0'}(v)}{2}-
\sum_{v\in V_{int}(T_0)-\{r\}} \binom{\kappa_{T_0}(v)}{2}\\ & \displaystyle =\binom{\kappa_{T'_0}(x)}{2}> 0. 
\end{array}
$$
\end{proof}

\begin{corollary}\label{cor:basic1}
For every non-binary phylogenetic tree $T\in \TT_n$, there always exists a binary phylogenetic tree $T'$ such that $\Phi(T')>\Phi(T)$.
\end{corollary}

\begin{proof}
Let $T\in \TT_n$ be a non-binary phylogenetic tree. Then it contains an internal node $z$ whose rooted subtree looks like the tree $T_0$ in the previous lemma, for some $k\geq 3$. By Lemma \ref{lem:basic} and the last lemma,  if  $T'\in\TT_n$ is the tree obtained from $T$ by replacing $T_0$ by $T_0'$ as its subtree rooted at $z$, then $\Phi(T')-\Phi(T)> 0$.
\end{proof}

Therefore, the maximum total cophenetic index is reached at a binary tree.

\begin{lemma}\label{lem:basic2}
Let $m\geq 4$, let $2\leq k\leq m-2$, let $T_1$ be any binary  tree on $\{k+1,\ldots,m\}$, and let  $T_0$ and $T_0'$ be the phylogenetic trees in $\TB_m$ depicted  in Fig.~\ref{fig:2}. Then,
$\Phi(T'_0)-\Phi(T_0)>0$.
\end{lemma}

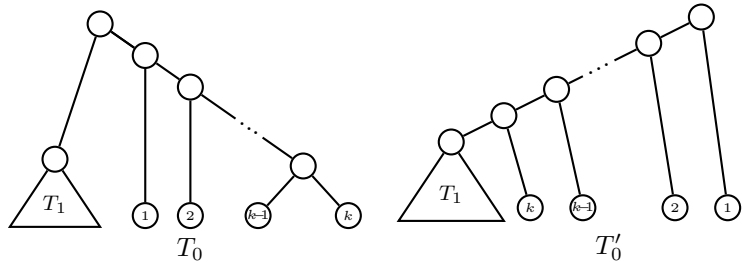
\begin{figure}[htb]
\begin{center}
\begin{tikzpicture}[thick,>=stealth,scale=0.3]
\draw(0,0) node[tre] (z) {}; 
\draw (z)--(-2,-3)--(2,-3)--(z);
\draw(0,-2) node  {\footnotesize $T_1$};
\draw (2,6) node[tre] (x) {}; 
\draw  (x)--(z);
\draw(4,4.6) node[tre] (a) {}; 
\draw (x)--(a);
\draw(6,3.2) node[tre] (b) {}; 
\draw (a)--(b);
\draw(4,-2.5) node [tre] (1) {}; \etq 1
\draw(6,-2.5) node [tre] (2) {}; \etq 2
\draw (b)--(8,1.8);
\draw (8.3,1.8-0.7*0.3) node {.};
\draw (8.6,1.8-0.7*0.6) node {.};
\draw (8.9,1.8-0.7*0.9) node {.};
\draw(11,-0.3) node[tre] (c) {}; 
\draw (9.2,1.8-0.7*1.2)--(c);
\draw(9,-2.5) node [tre] (k-1) {};
\draw (k-1) node {\tiny $k\!\!\!-\!\!\!1$};  
\draw(13,-2.5) node [tre] (k) {}; \etq{k}
\draw (x)--(a);
\draw (a)--(1);
\draw  (b)--(2);
\draw (c)--(k);
\draw (c)--(k-1);
\draw(6,-4) node  {$T_0$};
\end{tikzpicture}
\quad
\begin{tikzpicture}[thick,>=stealth,scale=0.35]
\draw(0,0) node[tre] (z) {}; 
\draw (z)--(-2,-3)--(2,-3)--(z);
\draw(0,-2) node  {\footnotesize $T_1$};
\draw(2,1) node[tre] (a) {}; 
\draw(3,-2.5) node [tre] (1) {}; 
\draw(1) node  {\tiny $k$};
\draw(4,2) node[tre] (b) {}; 
\draw(5,-2.5) node [tre] (2) {}; 
\draw(2) node  {\tiny $k\!\!\!-\!\!\!1$};
\draw (b)--(5,2.5);
\draw (5.3,2.65) node {.};
\draw (5.6,2.8) node {.};
\draw (5.9,2.95) node {.};
\draw(7.5,3.75) node[tre] (c) {}; 
\draw (6.2,3.1)--(c);
\draw(8.5,-2.5) node [tre] (k) {}; 
\draw(k) node  {\tiny $2$};
\draw(9.5,4.75) node[tre] (r) {}; 
\draw(10.5,-2.5) node [tre] (k0) {}; 
\draw(k0) node  {\tiny $1$};
\draw (r)--(k0);
\draw (r)--(c);
\draw  (a)--(z);
\draw (a)--(1);
\draw (b)--(a);
\draw  (b)--(2);
\draw (c)--(k);
\draw(6,-4) node  {$T_0'$};
\end{tikzpicture}

\end{center}
\caption{\label{fig:2} 
The trees $T_0$ and $T_0'$ in the statement of Lemma \ref{lem:basic2}.}
\end{figure}

\begin{proof}
By Lemma \ref{lem:alt1}, and recalling that $|L(T_1)|=m-k$, we have that 
$$
\begin{array}{l}
\displaystyle \Phi(T_0)=\phi(T_1)+\binom{m-k}{2}+\binom{k}{2}+\binom{k-1}{2}+\cdots+\binom{3}{2}+\binom{2}{2}\\
\displaystyle \Phi(T'_0)=\binom{m-1}{2}+\binom{m-2}{2}+\cdots+\binom{m-k+1}{2}+\binom{m-k}{2}+\phi(T_1)
\end{array}
$$
and hence, since $m-k\geq 2$,
$$
\Phi(T_0')-\Phi(T_0)=\sum_{j=1}^{k-1}\binom{m-k+j}{2}-\binom{j+1}{2}>0.
$$
\end{proof}

\begin{proposition}\label{ref:largest}
The trees in $\TT_n$ with maximum  total cophenetic index are exactly the rooted caterpillars $K_n$, and this maximum is
$\Phi(K_n)=\binom{n}{3}$.
\end{proposition}

\begin{proof}
By Corollary \ref{cor:basic1}, any tree in $\TT_n$ with maximum  total cophenetic index will be binary. Let now $T\in\TB_n$ and assume that it is not a caterpillar. Therefore, it has an internal node $z$ of largest depth without any leaf child; in particular, all internal descendant nodes of $z$ have some leaf child. Thus, and up to a relabeling of its leaves, the subtree of $T$ rooted at $z$ has the form of the tree $T_0$  in Fig. \ref{fig:larg1}, for some $k\geq 2$ and some $l\geq k+2$.  But then, by Lemma \ref{lem:basic2} (taking as $T_1$ the caterpillar subtree rooted at the parent of the leaf $k$),   the tree $T_0'$ also depicted in Fig. \ref{fig:larg1} has a strictly larger total cophenetic index. 
Then, by Lemma \ref{lem:basic}, if we replace in $T$ the subtree rooted at $z$ by this tree $T_0'$, we obtain a new tree $T'$ with $\Phi(T')> \Phi(T)$.
This implies that no tree other than a caterpillar can have the largest 
total cophenetic index.

\begin{figure}[htb]
\begin{center}
\begin{tikzpicture}[thick,>=stealth,scale=0.25]
\draw(0,0) node[tre] (1) {};   \etq 1
\draw(2,2) node[tre] (a) {}; 
\draw(3,0) node [tre] (2) {}; \etq 2
\draw(4,4) node[tre] (b) {}; 
\draw(5,2) node [tre] (3) {}; \etq 3
\draw (b)--(5,5);
\draw (5.3,5.3) node {.};
\draw (5.6,5.6) node {.};
\draw (5.9,5.9) node {.};
\draw(7.5,7.5) node[tre] (c) {}; 
\draw (6.2,6.2)--(c);
\draw(8.5,4.5) node [tre] (k) {}; \etq k
\draw(9.5,9.5) node[tre] (d) {}; 
\draw (d) node {\tiny $z$};  
\draw(19,0) node[tre] (l) {};   \etq l
\draw(17,2) node[tre] (a1) {}; 
\draw(16,0) node [tre] (l-1) {}; 
\draw (l-1) node {\tiny $l\!\!\!-\!\!\!1$};  
\draw(15,4) node[tre] (b1) {}; 
\draw(14,2) node [tre] (l-2) {}; 
\draw (l-2) node {\tiny $l\!\!\!-\!\!\!2$};  
\draw (b1)--(14,5);
\draw (13.7,5.3) node {.};
\draw (13.4,5.6) node {.};
\draw (13.1,5.9) node {.};
\draw(11.5,7.5) node[tre] (c1) {}; 
\draw (12.8,6.2)--(c1);
\draw(10.5,4.5) node [tre] (k+1) {}; 
\draw (k+1) node {\tiny $k\!\!\!+\!\!\!1$};  
\draw  (a)--(1);
\draw (a)--(2);
\draw (b)--(a);
\draw  (b)--(3);
\draw (c)--(k);
\draw (d)--(c);
\draw  (a1)--(l);
\draw (a1)--(l-1);
\draw (b1)--(a1);
\draw  (b1)--(l-2);
\draw (c1)--(k+1);
\draw (d)--(c1);
\draw(9.5,-2) node  {$T_0$};
\end{tikzpicture}
\quad
\begin{tikzpicture}[thick,>=stealth,scale=0.3]
\draw(0,0) node[tre] (1) {};   \etq 1
\draw(2,2) node[tre] (a) {}; 
\draw(4,0) node[tre] (2) {}; \etq 2
\draw(4,3) node[tre] (b) {}; 
\draw(6,0) node[tre] (3) {}; \etq 3
\draw (b)--(5,3.5);
\draw (5.3,3.65) node {.};
\draw (5.6,3.8) node {.};
\draw (5.9,3.95) node {.};
\draw(7.5,4.75) node[tre] (c) {}; 
\draw (6.2,4.1)--(c);
\draw(10,0) node[tre] (k) {}; \etq k
\draw(9.5,5.75) node[tre] (d) {}; 
\draw(12,0) node[tre] (k+1) {}; 
\draw (k+1) node {\tiny $l$};  
\draw(d)--(10,6); 
\draw (10.3,6.15) node {.};
\draw (10.6,6.3) node {.};
\draw (10.9,6.45) node {.};
\draw(12.5,7.25) node[tre] (e) {}; 
\draw (11.2,6.6)--(e);
\draw (e) node {\tiny $z$};  
\draw(16,0) node[tre] (l) {};   
\draw (l) node {\tiny $k\!\!\!+\!\!\!1$};  
\draw  (a)--(1);
\draw (a)--(2);
\draw (b)--(a);
\draw  (b)--(3);
\draw (c)--(k);
\draw (d)--(c);
\draw (d)--(k+1);
\draw (e)--(l);
\draw(9.5,-2) node  {$T_0'$};
\end{tikzpicture}
\end{center}
\caption{\label{fig:larg1} 
The trees $T_0$ and $T_0'$ in the proof of Proposition \ref{ref:largest}.}
\end{figure}
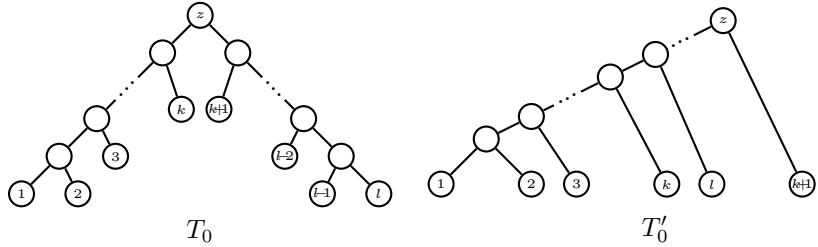

As far as the total cophenetic index of the rooted caterpillar $K_n$ with $n$ leaves depicted in Fig. \ref{fig:exs}.(a) goes, since the parent of the leaf labelled $j$, for $j=2,\ldots,n$, has $j$ descendant leaves, by Lemma \ref{lem:alt1} we have that
$\Phi(K_n)=\sum\limits_{j=2}^{n-1}\binom{j}{2}=\binom{n}{3}$. 
\end{proof} 

It is obvious that minimum total cophenetic index is 0, and it is attained only at the rooted star trees, depicted in Fig. \ref{fig:exs}.(b).  Therefore, the range of $\Phi$ on $\TT_n$ goes from 0 to $\binom{n}{3}$. This is one order of magnitude larger than the range of Sackin's and Colless' indices, whose maximum value, reached  also at the rooted caterpillars, has order $O(n^2)$  \cite{Heard92,Rogers:96,Shao:90}.

Let us characterize now those \emph{binary} phylogenetic trees with smallest total cophenetic index. 

\begin{lemma}\label{lem:min}
Let $T_1,T_2,T_3,T_4$ be an ordered binary forest on $\{1,\ldots,m\}$,  let
$x_i=|L(T_i)|$, for $i=1,2,3,4$, and assume that $x_1\geq x_2$, $x_3\geq x_4$ and $x_1> x_3$. Let $T_0$ the phylogenetic tree depicted in Fig \ref{fig:min1}.(a), and let $T\in \TB_n$ ($n\geq m$) be a binary phylogenetic tree having $T_0$ as a subtree rooted at some node. If $\Phi(T)$ is minimum in $\TB_n$, then $x_4\geq x_2$.
\end{lemma}

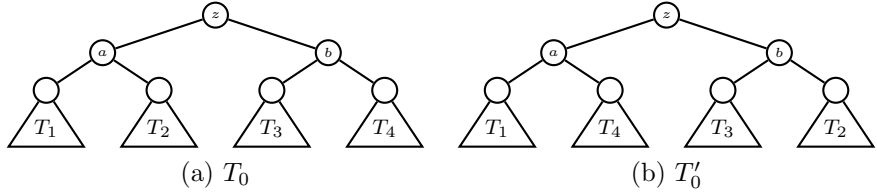
\begin{figure}[htb]
\begin{center}
\begin{tikzpicture}[thick,>=stealth,scale=0.25]
\draw(0,2) node[tre] (v1) {};   
\draw(6,2) node[tre] (v2) {}; 
\draw(12,2) node[tre] (v3) {}; 
\draw(18,2) node[tre] (v4) {}; 
\draw(3,4) node[tre] (a) {}; \etq a
\draw(15,4) node[tre] (b) {}; \etq b
\draw(9,6) node[tre] (z) {}; \etq z
\draw (z)--(a);
\draw (z)--(b);
\draw (a)--(v1);
\draw (a)--(v2);
\draw (b)--(v3);
\draw (b)--(v4);
\draw (v1)--(-2,-1)--(2,-1)--(v1);
\draw(0,0) node  {\footnotesize $T_1$};
\draw (v2)--(4,-1)--(8,-1)--(v2);
\draw(6,0) node  {\footnotesize $T_2$};
\draw (v3)--(10,-1)--(14,-1)--(v3);
\draw(12,0) node  {\footnotesize $T_3$};
\draw (v4)--(16,-1)--(20,-1)--(v4);
\draw(18,0) node  {\footnotesize $T_4$};
\draw(9,-2.5) node {(a) $T_0$};
\end{tikzpicture}
\quad
\begin{tikzpicture}[thick,>=stealth,scale=0.25]
\draw(0,2) node[tre] (v1) {};   
\draw(6,2) node[tre] (v2) {}; 
\draw(12,2) node[tre] (v3) {}; 
\draw(18,2) node[tre] (v4) {}; 
\draw(3,4) node[tre] (a) {}; \etq a
\draw(15,4) node[tre] (b) {}; \etq b
\draw(9,6) node[tre] (z) {}; \etq z
\draw (z)--(a);
\draw (z)--(b);
\draw (a)--(v1);
\draw (a)--(v2);
\draw (b)--(v3);
\draw (b)--(v4);
\draw (v1)--(-2,-1)--(2,-1)--(v1);
\draw(0,0) node  {\footnotesize $T_1$};
\draw (v2)--(4,-1)--(8,-1)--(v2);
\draw(6,0) node  {\footnotesize $T_4$};
\draw (v3)--(10,-1)--(14,-1)--(v3);
\draw(12,0) node  {\footnotesize $T_3$};
\draw (v4)--(16,-1)--(20,-1)--(v4);
\draw(18,0) node  {\footnotesize $T_2$};
\draw(9,-2.5) node {(b) $T'_0$};
\end{tikzpicture}
\end{center}
\caption{\label{fig:min1} 
(a) The tree $T_0$ in the statement of Lemma \ref{lem:min}.
(b) The tree $T_0'$ in the proof of Lemma \ref{lem:min}.}
\end{figure}

\begin{proof}
Assume that $x_2>x_4$. We shall show that, in this case, a suitable interchange of cousins in $T_0$ produces a tree with smaller total cophenetic index, which in particular will imply that $\Phi(T)$ cannot be the minimum in $\TB_n$.

Assume that the tree $T$ in the statement has the subtree $T_0$ rooted at a node $z$. Consider the tree $T_0'$ obtained by interchanging in $T_0$ the subtrees $T_2$ and $T_4$ (see Fig. \ref{fig:min1}.(b)) and let $T'$ be the tree obtained from $T$ by replacing $T_0$ by $T_0'$ as its subtree rooted at $z$. Then, by Lemma \ref{lem:alt1},
$$
\begin{array}{rl}
\Phi(T')-\Phi(T) & =\Phi(T_0')-\Phi(T_0)\\ & \displaystyle =\binom{\kappa_{T_0'}(a)}{2}+\binom{\kappa_{T_0'}(b)}{2}-\binom{\kappa_{T_0}(a)}{2}-\binom{\kappa_{T_0}(b)}{2}\\
& \displaystyle =\binom{x_1+x_4}{2}+\binom{x_2+x_3}{2}-\binom{x_1+x_2}{2}-\binom{x_3+x_4}{2}\\ \displaystyle 
& =x_1x_4+x_2x_3-x_1x_2-x_3x_4=(x_1-x_3)(x_4-x_2)<0
\end{array}
$$
which shows that $\Phi(T')<\Phi(T)$.
 \end{proof}
 
From the proof of the last lemma we deduce that if, in the tree  $T_0$ in Fig.~\ref{fig:min1}.(a), $|L(T_1)|\neq |L(T_3)|$ and $|L(T_2)|\neq |L(T_4)|$, and if we interchange $T_2$ and $T_4$, then the resulting tree has always a different total cophenetic index.

\begin{lemma}\label{lem:min2}
Let $T_1,T_2,$ be an ordered binary forest on $\{1,\ldots,m-1\}$,  let
$x_i=|L(T_i)|$, for $i=1,2$, and assume that $x_1\geq x_2$. Let $T_0$ the phylogenetic tree depicted in Fig \ref{fig:min3}.(a), and let $T\in \TB_n$ be a binary phylogenetic tree having $T_0$ as a subtree rooted at some node. If $\Phi(T)$ is minimum in $\TB_n$, then $x_1=x_2=1$.
\end{lemma}
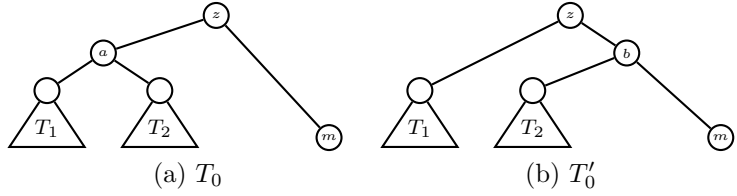
\begin{figure}[htb]
\begin{center}
\begin{tikzpicture}[thick,>=stealth,scale=0.25]
\draw(0,2) node[tre] (v1) {};   
\draw(6,2) node[tre] (v2) {}; 
\draw(3,4) node[tre] (a) {}; \etq a
\draw(15,-0.5) node[tre] (m) {}; \etq m
\draw(9,6) node[tre] (z) {}; \etq z
\draw (z)--(a);
\draw (z)--(m);
\draw (a)--(v1);
\draw (a)--(v2);
\draw (v1)--(-2,-1)--(2,-1)--(v1);
\draw(0,0) node  {\footnotesize $T_1$};
\draw (v2)--(4,-1)--(8,-1)--(v2);
\draw(6,0) node  {\footnotesize $T_2$};
\draw(7.5,-2.5) node {(a) $T_0$};
\end{tikzpicture}
\quad
\begin{tikzpicture}[thick,>=stealth,scale=0.25]
\draw(0,2) node[tre] (v1) {};   
\draw (v1)--(-2,-1)--(2,-1)--(v1);
\draw(0,0) node  {\footnotesize $T_1$};
\draw(8,6) node[tre] (z) {}; \etq z
\draw (z)--(v1);
\draw(6,2) node[tre] (v2) {}; 
\draw (v2)--(4,-1)--(8,-1)--(v2);
\draw(6,0) node  {\footnotesize $T_2$};
\draw(11,4) node[tre] (b) {}; \etq b
\draw(16,-0.5) node[tre] (m) {}; \etq m
\draw (z)--(b);
\draw (b)--(m);
\draw (b)--(v2);
\draw(7.5,-2.5) node {(b) $T'_0$};
\end{tikzpicture}
\end{center}
\caption{\label{fig:min3} 
(a) The tree $T_0$ in the statement of Lemma \ref{lem:min2}. (b) The tree $T_0'$ in the proof of Lemma \ref{lem:min2}.}
\end{figure}

\begin{proof}
Assume that $x_1>1$. We shall show that, again in this case, a suitable interchange of cousins in $T_0$ produces a tree with smaller total cophenetic index.

Assume that the tree $T$ in the statement has the subtree $T_0$ rooted at a node $z$.  Let $T'$ by the tree obtained from $T$ by replacing $T_0$ by the subtree $T_0'$ described in Fig. \ref{fig:min3}.(b). Then:
$$
\begin{array}{rl}
\Phi(T')-\Phi(T) & =\Phi(T_0')-\Phi(T_0)\\ & \displaystyle =\binom{\kappa_{T_0'}(b)}{2}-\binom{\kappa_{T_0}(a)}{2}=
\binom{x_2+1}{2}-\binom{x_1+x_2}{2}<0
\end{array}
$$
which shows that $\Phi(T')<\Phi(T)$. 
\end{proof}

The last two lemmas show that, unlike what happens with Sackin's and Colless' indices, any interchange of cousins that changes the balance of their grandparent always changes the total cophenetic index of a tree.

\begin{theorem}\label{th:min}
For every $T\in \TB_n$, $\Phi(T)$ is minimum on $\TB_n$ if, and only if, $T$ is maximally balanced.
\end{theorem}

\begin{proof}
Assume that $T\in \TB_n$ is not maximally balanced, and let $z$ be a non-balanced internal node in $T$  with largest depth. Assume that $a$ and $b$ are its children, with $\kappa_T(a)\geq \kappa_T(b)+2$.

If $b$ is a leaf, then, by Lemma \ref{lem:min2}, $\kappa_T(a)=2$ and therefore $\kappa_T(a)\not\geq 3=\kappa_T(b)+2$. Therefore, $a$ and $b$ are internal, and hence balanced. Let $T_0$ be the subtree of $T$ rooted at $z$, represented in Fig. \ref{fig:min1}.(a), and let $x_i=|L(T_i)|$, for $i=1,2,3,4$; without any loss of generality, we shall assume that $x_1\geq x_2$ and $x_3\geq x_4$ and thus, since $a$ and $b$ are balanced, $x_2=x_1$ or $x_1-1$ and $x_4=x_3$ or $x_3-1$. Then, $x_1+x_2=\kappa_T(a)\geq \kappa_T(b)+2=x_3+x_4+2$ implies that
$2x_1\geq 2x_3+1$, and hence that $x_1>x_3$.

%

Therefore, by Lemma \ref{lem:min}, if $\Phi(T)$ is minimum in $\TB_n$, it must happen that $x_1>x_3\geq x_4\geq x_2$. Since it forbids the equality $x_1=x_2$, it implies that $x_1=x_2+1$ and therefore $x_2=x_3=x_4$. But then $x_1+x_2=2x_2+1\not\geq x_3+x_4+2=2x_2+2$, against the assumption that $z$ is not balanced. 
\end{proof}

So, the only binary trees with minimum $\Phi$ are the maximally balanced.
Let us compute now this minimum value of $\Phi$ on $\BT_n$.  

\begin{lemma}\label{lem:recmin}
For every $n$, let $f(n)$ be the minimum of $\Phi$ on $\TB_n$. Then, $f(1)=f(2)=0$ and
$$
f(n)=f(\lceil n/2\rceil)+f(\lfloor n/2\rfloor)+\binom{\lceil n/2\rceil}{2}+\binom{\lfloor n/2\rfloor}{2},\quad \mbox{ for }n\geq 3.
$$
\end{lemma}

\begin{proof}
This recurrence for $f(n)$ is a direct consequence of Lemma \ref{lem:rec} and the fact that the root of a maximally balanced tree in $\TB_n$ is balanced and the subtrees rooted at their children are maximally balanced.
\end{proof}

\begin{proposition}
For every $n\geq 0$, let $a(n)$ is the highest power of 2 that divides $n!$. Then, for every $n\geq 1$,
$$
f(n)= \sum_{k=0}^{n-1} a(n).
$$
\end{proposition}

\begin{proof}
The sequence $(a(n))_n$ is sequence A011371 in Sloane's \textsl{On-Line Encyclopedia of Integer Sequences}  \cite{Sloane}, where we learn that it satisfies the recurrence 
$$
a(n)=\lfloor n/2\rfloor+ a(\lfloor n/2\rfloor).
$$
Let now  $(x(n))_n$ denote the sequence of partial sums of $(a(n))_n$, which is 
sequence A174605 in Sloane's Encyclopedia. Then, 
the sequence $(x(n))_n$ starts with $x(0)=x(1)=0$ and it satisfies the recurrence
$$
x(n)-x(n-1)=a(n)=\lfloor n/2\rfloor+ a(\lfloor n/2\rfloor)=\lfloor n/2\rfloor+ x(\lfloor n/2\rfloor)-x(\lfloor n/2\rfloor-1).
$$

We want to prove that $f(n+1)=x(n)$, for every $n\geq 0$. 
Since $f(1)=f(2)=0$, it remains to  check the equality
$$
f(n+1)-f(n)=\lfloor n/2\rfloor+ f(\lfloor n/2\rfloor+1)-f(\lfloor n/2\rfloor),\quad \mbox{for $n\geq 2$}.
$$
We prove this equality with the help of Lemma \ref{lem:rec} and by distinguishing four cases, depending on the residue of $n$ mod 4.
\begin{itemize}
\item If $n=4m$, then
$$
\begin{array}{rl}
f(n+1)-f(n) & = f(2m+1)+f(2m)+\binom{2m+1}{2}+\binom{2m}{2}\\ & \qquad -(f(2m)+f(2m)+\binom{2m}{2}+\binom{2m}{2})\\
& = f(2m+1)-f(2m)+\binom{2m+1}{2}-\binom{2m}{2}\\
& = f(2m+1)-f(2m)+2m\\
&= f(\lfloor n/2\rfloor+1)-f(\lfloor n/2\rfloor)+\lfloor n/2\rfloor
\end{array}
$$

\item If $n=4m+1$, then
$$
\begin{array}{rl}
f(n+1)-f(n) & = f(2m+1)+f(2m+1)+\binom{2m+1}{2}+\binom{2m+1}{2}\\ & \qquad -(f(2m+1)+f(2m)+\binom{2m+1}{2}+\binom{2m}{2})\\
& = f(2m+1)-f(2m)+\binom{2m+1}{2}-\binom{2m}{2}\\
& = f(2m+1)-f(2m)+2m\\
&= f(\lfloor n/2\rfloor+1)-f(\lfloor n/2\rfloor)+\lfloor n/2\rfloor
\end{array}
$$

\item If $n=4m+2$, then
$$
\begin{array}{rl}
f(n+1)-f(n) & = f(2m+2)+f(2m+1)+\binom{2m+2}{2}+\binom{2m+1}{2}\\ & \qquad -(f(2m+1)+f(2m+1)+\binom{2m+1}{2}+\binom{2m+1}{2})\\
& = f(2m+2)-f(2m+1)+\binom{2m+2}{2}-\binom{2m+1}{2}\\
& = f(2m+2)-f(2m+1)+2m+1\\
&= f(\lfloor n/2\rfloor+1)-f(\lfloor n/2\rfloor)+\lfloor n/2\rfloor
\end{array}
$$

\item If $n=4m+3$, then
$$
\begin{array}{rl}
f(n+1)-f(n) & = f(2m+2)+f(2m+2)+\binom{2m+2}{2}+\binom{2m+2}{2}\\ & \qquad -(f(2m+2)+f(2m+1)+\binom{2m+2}{2}+\binom{2m+1}{2})\\
& = f(2m+2)-f(2m+1)+\binom{2m+2}{2}-\binom{2m+1}{2}\\
& = f(2m+2)-f(2m+1)+2m+1\\
&= f(\lfloor n/2\rfloor+1)-f(\lfloor n/2\rfloor)+\lfloor n/2\rfloor
\end{array}
$$
\end{itemize}
This completes the proof. 
\end{proof}

In particular, this yields a new meaning and a new recurrence for sequence A174605 in Sloane's Encyclopedia.

\section{Expected value of $\Phi$ under the Yule model}

Let $\Phi_n$ be the random variable that chooses a tree $T\in \TB_n$ and computes its total cophenetic index $\Phi(T)$.  In this section we determine the expected value of $\Phi_n$ under the Yule model.
To do this, we shall make use of the following lemma, which can be useful to study the expected value under the Yule model of other  \emph{binary recursive tree shape statistics} in the sense of \cite{Matsen}.

\begin{lemma}\label{lem:YI}
Let $I$ be a mapping that associates to each phylogenetic tree a real number $\mathbb{R}$ satisfying the following two conditions:
\begin{enumerate}
\item[(a)] It is  invariant under tree isomorphisms  and relabelings of leaves.

\item[(b)] There exists a mapping $f:\NN\times \NN\to \RR$ such that, for every  phylogenetic trees $T,T'$ on disjoint sets of taxa $S,S'$, respectively,
$$
I(T\,\widehat{\ }\, T')=I(T)+I(T')+f(|S|,|S'|).
$$
\end{enumerate}
For every $n\geq 1$, let $I_n$ be the random variable that chooses a tree $T\in \TB_n$ and computes $I(T)$, and let $E_Y(I_n)$ be its expected value under the Yule model. Then,
$$
E_Y(I_n)=\frac{1}{n-1}\Big(2\sum_{k=1}^{n-1} E_Y(I_k)  + \sum_{k=1}^{n-1}f(k,n-k)\Big).
$$
\end{lemma}

\begin{proof}
First of all, notice that if $T_k\in \TT(S_k)$, with $S_k\subsetneq \{1,\ldots,n\}$ with $|S_k|=k$, and $T'_{n-k}\in\TT( \{1,\ldots,n\}\setminus S_k)$, then
$$
P_Y(T_k\widehat{\ }\, {}T'_{n-k})=\dfrac{2}{(n-1)\binom{n}{k}} P_Y(T_{k})P_Y(T'_{n-k})
$$
where $P_Y$ denotes the probability of a phylogenetic tree under the Yule model. This assertion is a direct consequence of the explicit probabilities of $T_k$, $T'_{n-k}$ and $T_k\widehat{\ }\, {}T'_{n-k}$ under the Yule model given in \S 2.2, and the fact that $V_{int}(T_k\widehat{\ }\, {}T'_{n-k})=V_{int}(T_{k})\cup V_{int}(T'_{n-k})\cup\{r\}$ (where $r$ denotes the root of $T_k\widehat{\ }\, {}T'_{n-k}$), these unions being disjoint.

Let us compute  now $E_Y(I_n)$ using its very definition:
$$
\begin{array}{l}
E_Y(I_n) \displaystyle =\sum_{T\in \TB_n} I(T)\cdot p_Y(T)
\\
\quad \displaystyle =\sum_{k=1}^{n-1}\sum_{S_k\subsetneq\{1,\ldots,n\}\atop |S_k|=k}
 \sum_{T_k\in \TB(S_k)}\sum_{T'_{n-k}\in \TB(S_k^c)}I(T_k\widehat{\ }\, {}T'_{n-k})\cdot p_Y(T_k\widehat{\ }\, {}T'_{n-k})
 \\ 
\quad  \displaystyle  =\frac{1}{2} \sum_{k=1}^{n-1}\binom{n}{k}
 \sum_{T_k\in \TB_k}\sum_{T'_{n-k} \in \TB_{n-k}}(I(T_k)+I(T'_{n-k})\\
\quad  \displaystyle  \qquad\qquad\qquad\qquad\qquad +f(k,n-k))\cdot \frac{2}{(n-1)\binom{n}{k}} P_Y(T_{k})P_Y(T'_{n-k})\\
\quad  \displaystyle  =\frac{1}{n-1}\sum_{k=1}^{n-1} 
 \sum_{T_k}\sum_{T'_{n-k}}(I(T_k)+I(T'_{n-k})+f(k,n-k))P_Y(T_{k})P_Y(T'_{n-k})
 \\
\quad  \displaystyle  =\frac{1}{n-1}\sum_{k=1}^{n-1} 
\Big( \sum_{T_k}\sum_{T'_{n-k} }I(T_k)P_Y(T_{k})P_Y(T'_{n-k})   \\
\quad  \displaystyle\qquad\qquad\qquad\qquad + \sum_{T_k}\sum_{T'_{n-k} } I(T'_{n-k}) P_Y(T_{k})P_Y(T'_{n-k}) \\
\quad  \displaystyle\qquad\qquad\qquad\qquad
 + \sum_{T_k}\sum_{T'_{n-k}} f(k,n-k)P_Y(T_{k})P_Y(T'_{n-k}) \Big)
 \\
\quad  \displaystyle  =\frac{1}{n-1}\sum_{k=1}^{n-1} 
\Big( \sum_{T_k} I(T_k)P_Y(T_{k})  +  \sum_{T'_{n-k}} I(T'_{n-k})  P_Y(T'_{n-k})   + f(k,n-k) \Big)
 \\
\quad  \displaystyle  =\frac{1}{n-1}\sum_{k=1}^{n-1} 
(E_Y(I_k) + E_Y(I_{n-k})   + f(k,n-k))
 \\
\quad  \displaystyle  =\frac{1}{n-1}\Big(2\sum_{k=1}^{n-1} 
E_Y(I_k)  + \sum_{k=1}^{n-1}f(k,n-k)\Big)
\end{array}
$$
\end{proof}

\begin{theorem}
Under the Yule model, the expected value of $\Phi_n$ is
$$
E_Y(\Phi_n)=n(n-1)-2n\sum_{i=2}^n\frac{1}{i}
$$
\end{theorem}

\begin{proof}
Lemma \ref{lem:rec} implies that $\Phi$ satisfies the hypothesis of Lemma \ref{lem:YI} with $f(k,n-k)=\binom{k}{2}+\binom{n-k}{2}$.
Therefore
$$
\sum_{k=1}^{n-1}f(k,n-k)=\sum_{k=1}^{n-1}\Big(\binom{k}{2}+\binom{n-k}{2}\Big)=
2\sum_{k=1}^{n-1} \binom{k}{2},
$$
and hence
$$
E_Y(\Phi_n)=\frac{2}{n-1}\Big(\sum_{k=1}^{n-1} 
E_Y(\Phi_k)  + \sum_{k=1}^{n-1}\binom{k}{2}\Big).
$$
Then,
$$
\begin{array}{rl}
E_Y(\Phi_n) & \displaystyle = 
\frac{2}{n-1} \Big(\sum_{k=1}^{n-1} E_Y(\Phi_k) + \sum_{k=1}^{n-1}\binom{k}{2}\Big) \\
& \displaystyle = 
\frac{2}{n-1} E_Y(\Phi_{n-1})+\frac{n-2}{n-1}\cdot \frac{2}{n-2} \sum_{k=1}^{n-2} E_Y(\Phi_k)\\
&  \displaystyle\qquad\qquad\qquad    
+\frac{2}{n-1}\binom{n-1}{2}+\frac{n-2}{n-1}\cdot \frac{2}{n-2} \sum_{k=1}^{n-2} \binom{k}{2}
\\
& \displaystyle = 
\frac{2}{n-1} E_Y(\Phi_{n-1})+\frac{n-2}{n-1}E_Y(\Phi_{n-1})+\frac{2}{n-1}\binom{n-1}{2}\\
& \displaystyle
= \frac{n}{n-1} E_Y(\Phi_{n-1})+n-2
\end{array}
$$
To solve this equation, rewrite it as
$$
\frac{1}{n}E_Y(\Phi_n)=\frac{1}{n-1} E_Y(\Phi_{n-1})+\frac{n-2}{n}
$$
Setting $x_n=E_Y(\Phi_n)/n$, the sequence $(x_n)_n$ satisfies
$$
x_n=x_{n-1}+\frac{n-2}{n}\mbox{, starting with $x_2=0$}.
$$
Therefore 
$$
x_n=\sum_{i=3}^n \frac{i-2}{i}=(n-2)-2\sum_{i=3}^n\frac{1}{i}=(n-1)-2\sum_{i=2}^n\frac{1}{i}
$$
and thus, finally,
$$
E_Y(\Phi_n)=nx_n=n(n-1)-2n\sum_{i=2}^n\frac{1}{i}.
$$
\end{proof}

Let $S_n$ stand for the random variable that chooses a tree $T_n\in \TB_n$ and computes its Sackin index $S(T_n)$; cf. \S 2.2.
Notice that, since $E_Y(S_n)=2n\sum\limits_{j=2}^{n}1/j$ \cite{KiSl:93}, we have that
$$
E_Y(\Phi_n)+E_Y(S_n)=n(n-1).
$$
We have not been able to find a direct reason for this equality.

\begin{corollary}
$E_Y(\Phi_n)=n^2+(1-2 \gamma)n -2 n \ln(n)+o(n)$.
\end{corollary}

So, the order $O(n^2)$ of the expected value under the Yule model of the total cophenetic index on $\BT_n$ is larger than the order $O(n\log(n))$ of the expected values of Sackin's and Colless' indices \cite{BFJ:06}.

From the expected values of the Sackin and the total cophenetic indices, we can deduce the expected value of the total area $D$ on $\TB_n$ under the Yule model. 

\begin{corollary}
Let $D_n$ be the random variable that chooses a tree $T\in \TB_n$ and computes its total area $D(T)$.  Under the Yule model, its expected value is
$$
E_Y(D_n)=2n(n+1)\sum_{i=2}^n\frac{1}{i}-2n(n-1)
$$
\end{corollary}

\begin{proof}
From Lemma \ref{lem:total} we deduce that
$$
2\Phi_n+D_n=(n-1)S_n,
$$
and therefore
$$
E_Y(D_n =(n-1)E_Y(S_n)-2E_Y(\Phi_n) =2n(n+1)\sum_{i=2}^n\frac{1}{i}-2n(n-1).
$$
\end{proof}

\begin{remark}
In \cite[p. 143, eq. (35)]{Mul11}, it is claimed that
$$
E_Y(D_n)=2n(n+1)\sum_{i=2}^n\frac{1}{i}-\frac{5}{2}n(n-1),
$$
which cannot be correct: since all three trees $T\in \BT_3$ have $D(T)=8$, it must happen that $E_Y(D_3)=8$, while the expression given in \textsl{loc. cit.} yields 
$E_Y(D_3)=5$. And incidentally, our formula does yield the correct value in this case.
\end{remark}

\section{Expected value of $\Phi$ under the uniform model}

 In this section we determine the expected value of $\Phi_n$ under the uniform model.
This expected value of $\Phi_n$ will be easily deduced, through Lemma \ref{lem:total}, from the expected value of  the total area, which was obtained in \cite{MirR10}, and  the expected value of the Sackin index, which we obtain in Theorem \ref{th:sackin} below. This last formula is, to our knowledge, new.

Since, under the uniform model, all trees in $\TT_n$ have the same probability, $1/(2n-3)!!$,  the expected value of $S_n$  under the uniform model is
$$
E_U(S_n)=\frac{\sum_{T\in \TB_n} S(T)}{(2n-3)!!}.
$$
So, we need to compute the numerator in this fraction. 
\begin{lemma}
For every $n\geq 3$, $\displaystyle \sum_{T\in \TB_n} S(T)=n\sum_{k=1}^{n-1} \dfrac{(2n-k-3)!k^2}{(n-k-1)!2^{n-k-1}}$.
\end{lemma}

\begin{proof}
For every $k=1,\ldots,n-1$, let
$$
c_{k,n}=|\{ T\in \TB_n\mid \delta_T(1)=k\}|=|\{ T\in \TT_n\mid \delta_T(i)=k\}|\mbox{ for every $1\leq i\leq n$}.
$$
Then
$$
\begin{array}{rl}
\displaystyle \sum_{T\in \TB_n} S(T) & \displaystyle =\sum_{T\in \TB_n}\sum_{i=1}^{n} \delta_T(i) =\sum_{i=1}^{n}\sum_{T\in \TB_n} \delta_T(i)\\
& \displaystyle  =\sum_{i=1}^{n}\sum_{k=1}^{n-1} k\cdot |\{ T\in \TT_n\mid \delta_T(i)=k\}|\\
& \displaystyle =\sum_{i=1}^{n}\sum_{k=1}^{n-1} k\cdot |\{ T\in \TT_n\mid \delta_T(1)=k\}|  =n\sum_{k=1}^{n-1} k\cdot c_{k,n}.
\end{array}
$$
It remains to compute $c_{k,n}$ for $k\geq 1$. To do so, notice that every tree $T\in \TB_n$  such that $\delta(1)=k$ will have the form described in Fig. \ref{fig:forests}. Therefore, it is determined by the ordered $k$-forest $T_1,T_2,\ldots,T_{k}$ on $\{2,\ldots,n\}$, and thus
$$
c_{k,n}=|\FF_{k,n-1}|= \frac{(2n-k-3)!k}{(n-k-1)!2^{n-k-1}},
$$
from which the expression in the statement follows.
\end{proof}

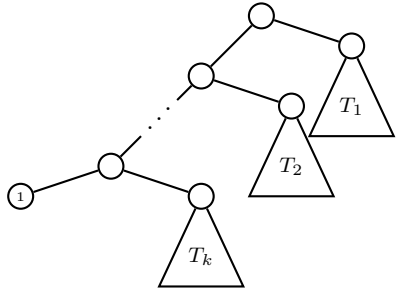
\begin{figure}[htb]
\begin{center}
\begin{tikzpicture}[thick,>=stealth,scale=0.4]
\draw(-1,3) node[tre] (1) {}; \etq 1
\draw(2,4) node[tre] (x1) {}; 
\draw (x1)--(1);
\draw (x1)--(3,5);
\draw (3.3,5.3) node {.};
\draw (3.6,5.6) node {.};
\draw (3.9,5.9) node {.};
\draw(5,7) node[tre] (x2) {}; 
\draw (4.2,6.2)--(x2);
\draw(7,9) node[tre] (r) {}; 
\draw (x2)--(r);
\draw(5,3) node[tre] (tk) {}; 
\draw (tk)--(3.6,0)--(6.4,0)--(tk);
\draw(5,1) node  {\footnotesize $T_{k}$};
\draw (x1)--(tk);
\draw(8,6) node[tre] (t2) {}; 
\draw (t2)--(6.6,3)--(9.4,3)--(t2);
\draw(8,4) node  {\footnotesize $T_{2}$};
\draw (x2)--(t2);
\draw(10,8) node[tre] (t1) {}; 
\draw (t1)--(8.6,5)--(11.4,5)--(t1);
\draw(10,6) node  {\footnotesize $T_{1}$};
\draw (r)--(t1);
\end{tikzpicture}
\end{center}
\caption{\label{fig:forests} 
The structure of a tree $T$ with $\delta_T(1)=k$.}
\end{figure}

Now, recall that the (\emph{generalized}) \emph{hypergeometric function} ${}_pF_q$ is defined  \cite{Bayley} as 
$$
{}_pF_q\bigg(\begin{array}{rrr} a_1,&\ldots, &a_p \\[-0.5ex] b_1,& \ldots,  &b_q\end{array};z\bigg)=\sum_{k\geq 0} \frac{(a_1)_k\cdots (a_p)_k}{(b_1)_k\cdots (b_q)_k}\cdot \frac{z^k}{k!},
$$ 
where $(a)_k := a\cdot (a+1)\cdots (a+k-1)$. Many popular software systems, like Mathematica or R, have implementations of these functions.

\begin{theorem}\label{th:sackin}
The expected value of the random variable $S_n$ under the uniform model  is
$$
E_U(S_n)=\frac{n}{2n-3} {}_3 F_2
\bigg(\begin{array}{l} 2,\ 2,\ 2-n \\[-0.5ex] 1,\  4-2n\end{array};2\bigg)
$$
\end{theorem}

\begin{proof}
By the last lemma, we have that
$$
E_U(S_n)=\frac{\sum_{T\in \TB_n} S(T)}{(2n-3)!!}=\frac{n}{(2n-3)!!} \sum_{k=1}^{n-1}  \frac{(2n-k-3)!k^2}{(n-k-1)!2^{n-k-1}}.
$$
Now
$$
\begin{array}{rl}
\displaystyle \frac{nk^2(2n-k-3)!}{(2n-3)!!(n-k-1)!2^{n-k-1}} & \displaystyle =
  \frac{nk^2(2n-k-3)!2^{n-2}(n-2)!}{(2n-3)!(n-k-1)!2^{n-k-1}}\\[2ex] & \displaystyle=
  \frac{nk^2(2n-k-3)!2^{n-2}(n-2)!k!}{(2n-3)!(n-k-1)!2^{n-k-1}k!} \\[2ex] & \displaystyle=
 \frac{nk^22^{k-1}\binom{n-1}{k}}{(n-1)\binom{2n-3}{k}}
\end{array}
$$
and thus
$$
\begin{array}{l}
E_U(S_n)  = \displaystyle
\frac{n}{n-1}\sum_{k=1}^{n-1}  k^22^{k-1}\cdot \frac{\binom{n-1}{k}}{\binom{2n-3}{k}}\\
\qquad = \displaystyle\frac{n}{n-1} \sum_{k=1}^{n-1} \frac{k^22^{k-1}(n-1)(n-2)(n-3) \cdots (n-k)}{(2n-3)(2n-4)(2n-5)\cdots (2n-k-2)}\\
\qquad = \displaystyle\frac{n}{2n-3} \sum_{k=1}^{n-1} \frac{k^22^{k-1}(n-2)(n-3) \cdots (n-k)}{(2n-4)(2n-5)\cdots (2n-k-2)}\\
\qquad = \displaystyle\frac{n}{2n-3} \sum_{k=1}^{n-1} \frac{k^22^{k-1}(2-n)(2-n+1)\cdots (-n+k)}{(4-2n)(4-2n+1)\cdots (2-2n+k)}\\
\qquad = \displaystyle\frac{n}{2n-3} \sum_{k=0}^{n-2} \frac{(k+1)^22^k(2-n)(2-n+1)\cdots (1-n+k)}{(4-2n)(4-2n+1)\cdots (3-2n+k)}\\
\qquad = \displaystyle\frac{n}{2n-3} \sum_{k\geq 0} \frac{((k+1)!)^2(2-n)(2-n+1)\cdots(1-n+k)\cdot 2^k}{(k!)^2(4-2n)(4-2n+1)\cdots (3-2n+k)}\\
\qquad = \displaystyle\frac{n}{2n-3} \sum_{k\geq 0} \frac{(2)_k(2)_k(2-n)_k}{(1)_k(4-2n)_k}\cdot \frac{2^k}{k!}  =\dfrac{n}{2n-3} {}_3 F_2
\bigg(\begin{array}{l} 2,\ 2,\ 2-n \\[ -0.5ex] 1,\  4-2n\end{array};2\bigg)
\end{array}
$$
as we claimed.
\end{proof}

We have now the following result.

\begin{theorem}
Under the uniform model, the expected value of $\Phi_n$ is
$$
E_U(\Phi_n)=\binom{n}{2}\Bigg(\frac{1}{2n-3} {}_3 F_2
\bigg(\begin{array}{l} 2,\ 2,\ 2-n \\[-0.5ex] 1,\  4-2n\end{array}; 2\bigg)-\frac{1}{2}\cdot \frac{(2n-2)!!}{(2n-3)!!}\Bigg)\sim \frac{\sqrt{\pi}}{4}n^{5/2}
$$
\end{theorem}

\begin{proof}
The expected values under the uniform model of $S_n$ and $D_n$ are:
$$
\begin{array}{ll}
\displaystyle E_U(S_n)=\frac{n}{2n-3} {}_3 F_2
\bigg(\begin{array}{l} 2,\ 2,\ 2-n \\[-0.5ex] 1,\  4-2n\end{array};2\bigg)&  \quad  \mbox{ by Theorem \ref{th:sackin}}
\\
\displaystyle E_U(D_n)=\binom{n}{2} \cdot\frac{(2n-2)!!}{(2n-3)!!}&   \quad  \mbox{\cite {MirR10}}
\end{array}
$$
Then, by Lemma \ref{lem:total},
$$
\begin{array}{rl}
E_U(\Phi) & \displaystyle =\frac{n-1}{2}E_U(S_n)-\frac{1}{2}E_U(D_n)\\ & \displaystyle =
\frac{n-1}{2}\cdot \frac{n}{2n-3} {}_3 F_2
\bigg(\begin{array}{l} 2,\ 2,\ 2-n \\[-0.5ex] 1,\  4-2n\end{array};2\bigg)
-\frac{1}{2}\binom{n}{2} \cdot\frac{(2n-2)!!}{(2n-3)!!}\\
& \displaystyle =
\binom{n}{2}\Bigg(\frac{1}{2n-3} {}_3 F_2
\bigg(\begin{array}{l} 2,\ 2,\ 2-n \\[-0.5ex] 1,\  4-2n\end{array}; 2\bigg)-\frac{1}{2}\cdot \frac{(2n-2)!!}{(2n-3)!!}\Bigg)
\end{array}
$$
The assertion $E_U(\Phi_n) \sim \frac{\sqrt{\pi}}{4}n^{5/2}$ comes easily from Lemma \ref{lem:total}, and the facts  that 
$E(S_n) \sim \sqrt{\pi}n^{3/2}$ \cite{BFJ:06}, and that, using Stirling's approximation for large factorials,
$E_U(D_n)\sim \frac{\sqrt{\pi}}{2}n^{5/2}$ \cite {MirR10}.  
\end{proof}

\section{Discussion and conclusions}
In this paper we have introduced a new balance index for phylogenetic trees, the total cophenetic index $\Phi$. This index makes sense for arbitrary phylogenetic trees, it can be computed in linear time, and it has a larger range of values than Sackin's or Colless' indices. We have computed its maximum and minimum values for binary and arbitrary phylogenetic trees, and its expected value under the Yule and the uniform models.  In a future work we plan to study other statistical properties of $\Phi$, like its variance, its limiting distribution or its correlation to other balance indices.

From the point of view of the measurement of the degree of symmetry of a tree, our index outperforms the resolution power of Sackin's and Colless' indices. We already saw some hints of this property in the previous sections: for instance, in Theorem \ref{th:min}, where we proved that the only trees $T\in \TB_n$ that have minimum $\Phi(T)$ are the maximally balanced, something that is not true in general for  Sackin's and Colless' indices (recall Fig. \ref{fig:minind}); or in the lemmas previous to the proof of this theorem, where we saw that any interchange of cousins that modifies the balance of their grandparent also modifies the value of $\Phi$. As a further evidence of this greater resolution power, we have estimated the probability that a pair of trees $T_1,T_2\in \BT_n$ have $I(T_1)=I(T_2)$, for $I=C,S,\Phi$. To do so, for every $n=2,\ldots,10^4$ we have chosen randomly a number $N$ of pairs of trees in $\TB_n$ (for the first few values of $n$, $N$ was taken to be $|\TB_n|$, but starting at $n=8$, we took $N=3000$), and computed, for $I=C,S,\Phi$,
$$
\hat{p}_n(I)=\frac{\mbox{number of pairs $(T_1,T_2)$ with $n$ leaves such that $I(T_1)=I(T_2)$}}{N}.
$$
Fig. \ref{TIES} summarizes the results. It plots $\log(\hat{p}_n(I))$ for the three balance indices  as a function of $\log(n)$. We can see that that the total cophenetic index has the lowest such estimated probability of a tie.   We plan to perform a deeper study of the probability of ties for the different balance indices in a future paper.

\begin{figure}[htb]
\begin{center}
\begin{minipage}{0.6\linewidth}
\begin{center}
\includegraphics[height=6cm]{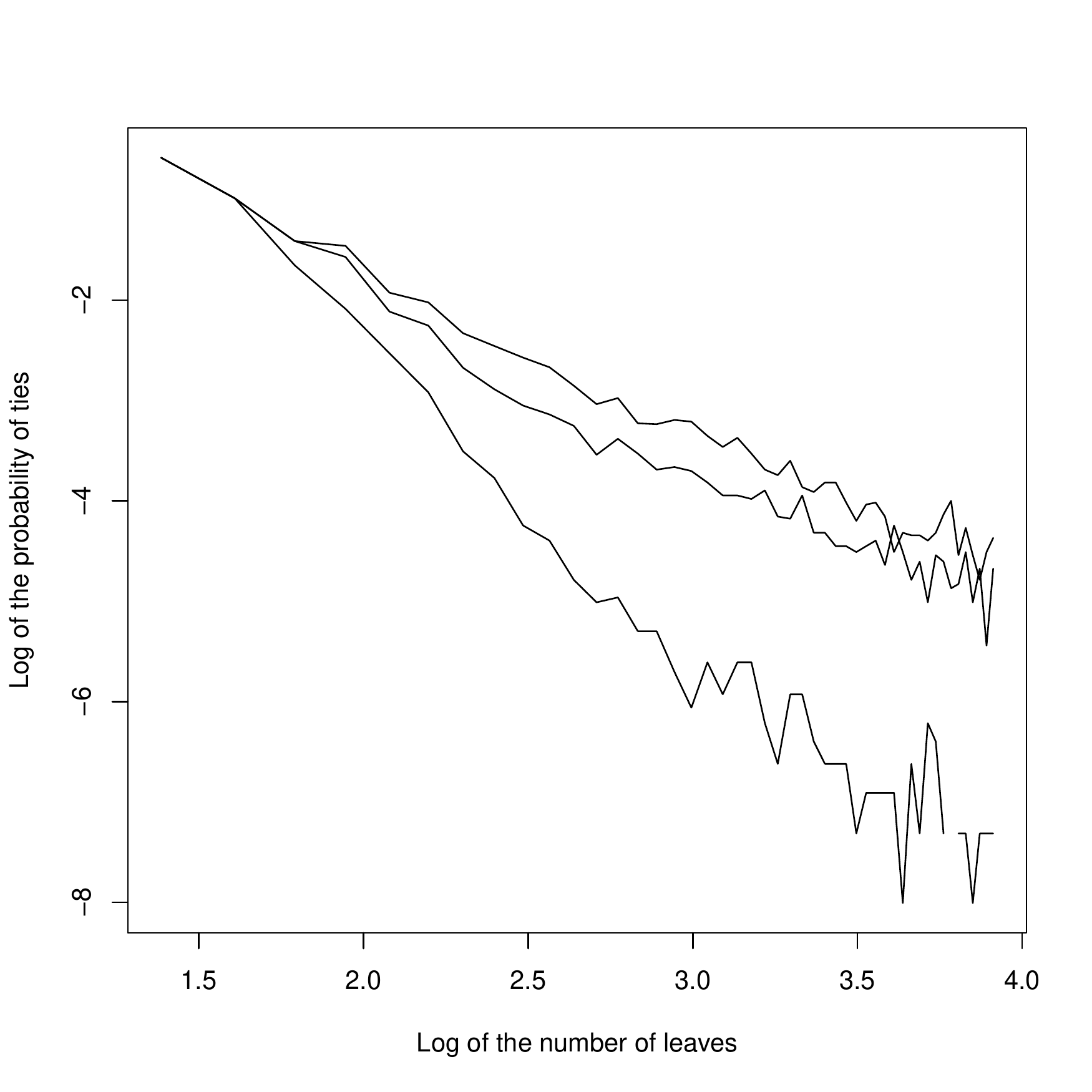}
\end{center}
\end{minipage}
\hspace*{-0.75cm}\begin{minipage}{0.2\linewidth}
\vspace*{1.5cm}

$\leftarrow$ Sackin

$\leftarrow$  Colless
\vspace*{0.7cm}

$\leftarrow$ $\Phi$
\end{minipage}
\end{center}\label{TIES}%
\caption{Log-log plot of the estimated probability of a tie for three balance indices.}
\end{figure}%

This greater resolution power of $\Phi$ makes it a better candidate to be used to test evolutionary hypotheses.
We have performed a preliminary such test on the TreeBASE database \cite{treebase}. We have considered the numbers $n$ of leaves for which the TreeBASE contains at least 20 binary phylogenetic trees with $n$ leaves, and for each such $n$ we have computed the mean of the total cophenetic indices of the corresponding binary trees. Fig. \ref{means} plots the log of these means as a function of  $\log(n)$. We have added the curves of the log of the expected values of $\Phi_n$ under the Yule distribution (lower curve) and under the uniform distribution (upper curve), again as a function of  $\log(n)$. This figure shows that the total cophenetic indices of the binary phylogenetic trees in TreeBASE are better explained by the uniform model than by the Yule  model. We also plan to report  in a future paper on more extensive tests on stochastic models of evolutionary processes using the total cophenetic index.

\begin{figure}[htb]
\centering
\includegraphics[height=6cm]{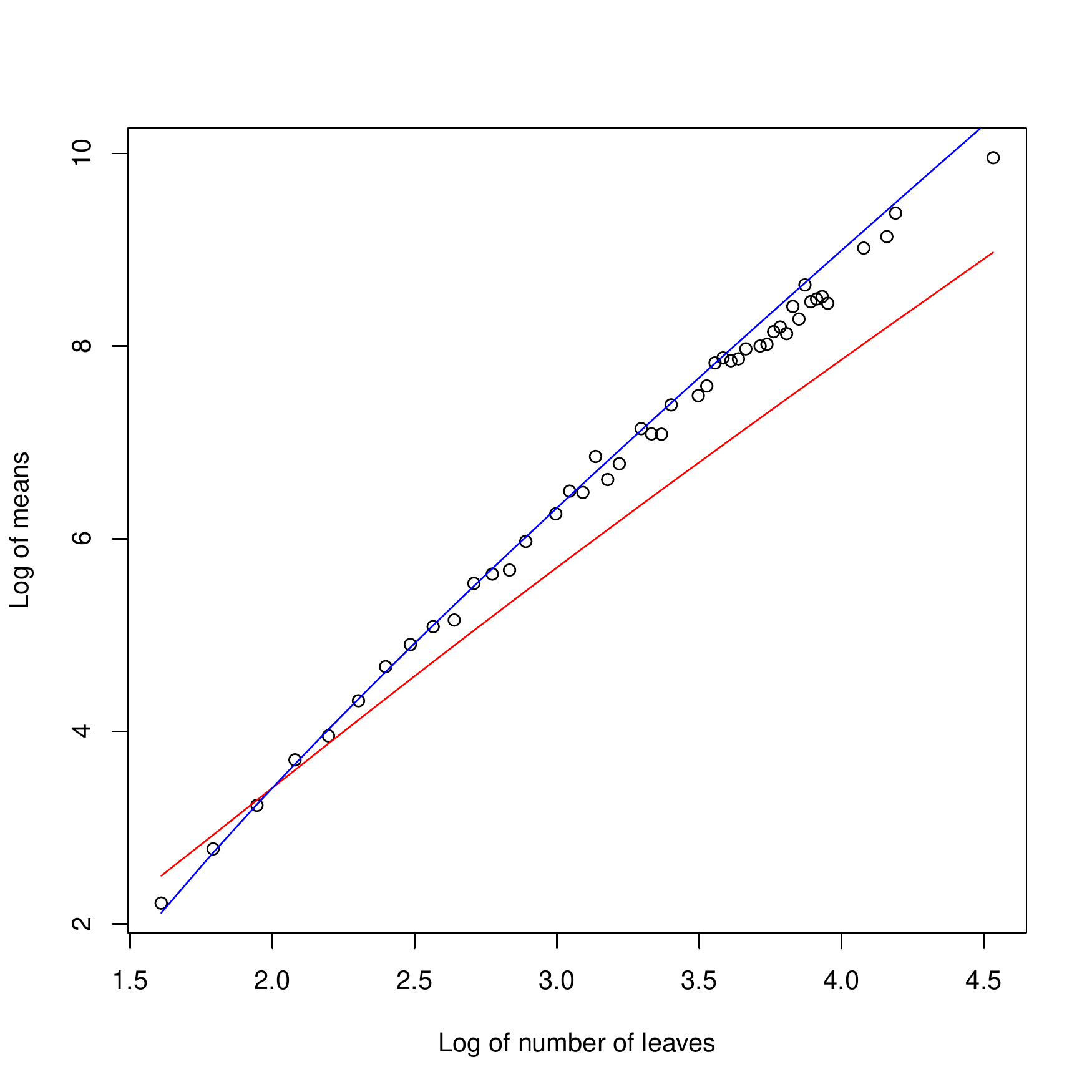}
\label{means}%
\caption{Log-log plots of the mean of the total cophenetic index of the binary trees in TreeBASE with a fixed number  $n$ of leaves,  of $E_Y(\Phi_n)$ (lower curve) and $E_U(\Phi_n)$ (upper curve).}
\end{figure}%

\section*{Acknowledgements}  The research reported in this paper has been partially supported by the Spanish government and the UE FEDER program, through projects MTM2009-07165 and TIN2008-04487-E/TIN.  We thank G. Cardona, E. Hern\'andez-Garc\'\i a, and J. Mir\'o for several comments on previous versions of this work.

%

\begin{thebibliography}{1}
\expandafter\ifx\csname url\endcsname\relax
  \def\url#1{\texttt{#1}}\fi
\expandafter\ifx\csname urlprefix\endcsname\relax\def\urlprefix{URL }\fi

\bibitem{Bayley}
 W.N. Bayley,  Generalized Hypergeometric Series. Cambridge Tracts in Mathematics and Mathematical Physics 32, Stechert-Hafner Service Inc. (1964).
 
\bibitem{BF:05}
M. G. B. Blum, O. Fran\c cois,
On statistical tests of phylogenetic tree imbalance: The Sackin and other indices revisited.
Mathematical Biosciences, 195 (2005), 141--153.

\bibitem{BFJ:06} M. G. B. Blum, O. Fran\c cois, S. Janson,
The mean, variance and limiting distribution of two statistics sensitive to phylogenetic tree balance. Ann. Appl. Probab.  16 (2006), 2195--2214. 

\bibitem{Brown} J. Brown, Probabilities of evolutionary trees. Syst. Biol. 43 (1994), 78--91.

\bibitem{CS} L. L. Cavalli-Sforza,  A. Edwards, Phylogenetic analysis. Models and estimation procedures. Am. J. Hum. Genet., 19 (1967), 233--257.


\bibitem{Colless:82} D. H. Colless, Review of ``Phylogenetics: the theory and practice of phylogenetic systematics''. Sys. Zool, 31 (1982), 100--104.

\bibitem{fel:04} J.~Felsenstein,  Inferring Phylogenies. Sinauer Associates Inc., 2004.


\bibitem{Harding71} E. Harding, The probabilities of rooted tree-shapes generated by random bifurcation. Adv. Appl. Prob. 3  (1971), 44--77.

\bibitem{Heard92} S. B. Heard,
Patterns in Tree Balance among Cladistic, Phenetic, and Randomly Generated Phylogenetic Trees.
Evolution
46  (1992),  1818--1826 

\bibitem{KiSl:93}
M. Kirkpatrick, M. Slatkin, Searching for evolutionary patterns in the shape
of a phylogenetic tree. Evolution 47 (1993), 1171--1181.

\bibitem{Matsen}
F. Matsen, Optimization Over a Class of Tree Shape Statistics.
IEEE/ACM Trans. Comput. Biol. Bioinformatics, 4 (2007), 
506--512.


\bibitem{MirR10} 
A. Mir, F. Rossell\'o, The mean value of the squared path-difference distance
  for rooted phylogenetic trees.
Journal of Mathematical Analysis and Applications  371 (2010),
  168--176.
  

\bibitem{Mooers97}
A. Mooers, S. B. Heard, Inferring evolutionary process from phylogenetic tree shape. Quart. Rev. Biol. 72  (1997) 31--54.


\bibitem{treebase} V. Morell,  TreeBASE: the roots of phylogeny. Science 273 (1996), 569--560.
\url{http://www.treebase.org}

\bibitem{Mul11}  W. H. Mulder,
Probability distributions of ancestries and genealogical distances on stochastically generated rooted binary trees. J. Theor. Biol. 280 (2011), 139--145.

\bibitem{Rogers:93} J. S. Rogers, Response of tree imbalance to number of terminal taxa. Sys. Biol. 42 (1993), 102--105.

\bibitem{Rogers:94} J. S. Rogers, Central moments and probability distributions of Colless's coefficient of tree imbalance. Evolution 48 (1994), 2026--2036.

\bibitem{Rogers:96} J. S. Rogers, Central moments and probability distributions of three measures of phylogenetic tree imbalance, Sys. Biol. 45 (1996), 99--110.


\bibitem{Rosen78} D. E. Rosen, Vicariant Patterns and Historical Explanation in Biogeography.
Syst. Biol. 27 (1978), 159--188.

  
\bibitem{Sackin:72} M. J. Sackin, ``Good'' and ``bad'' phenograms. Sys. Zool, 21 (1972), 225--226.

\bibitem{Shao:90} K.T. Shao, R. Sokal, Tree balance. Sys. Zool, 39 (1990), 226--276.

\bibitem{Sokal:62} R. Sokal, F. Rohlf, 
The Comparison of Dendrograms by Objective Methods.
Taxon 11 (1962),  33--40.


\bibitem{cherries} M. Steel, A. McKenzie, Distributions of cherries for two models of trees. Math. Biosc. 164 (2000), 81--92.

\bibitem{SM01} M. Steel, A. McKenzie, Properties of phylogenetic trees generated by Yule-type speciation models. Math. Biosc. 170 (2001), 91--112.

\bibitem{Sloane}
{The On-Line Encyclopedia of Integer Sequences} (2010), published electronically at \url{http://oeis.org/}.

\bibitem{Val:02}
G. Valiente, Algorithms on Trees and Graphs. Springer-Verlag (2002).

\bibitem{Yule} G. U. Yule, A mathematical theory of evolution based on the conclusions of Dr J. C. Willis. Phil. Trans.   Royal Soc. (London) Series B 213 (1924), 21--87.

\end{thebibliography}

\end{document}